\ifdefined\XeTeXversion\else\pdfoutput=1\fi
\documentclass[12pt]{article}

\usepackage{graphicx}
\usepackage{multirow}
\usepackage{amsmath,amssymb,amsfonts,amsthm,mathtools}
\usepackage{mathrsfs}
\usepackage{xcolor}
\usepackage{manyfoot,natbib}
\usepackage{booktabs,threeparttable,tabularx}
\usepackage{enumitem}
\usepackage{float}
\usepackage{lmodern,anyfontsize}
\usepackage{bm}
\usepackage{fullpage}
\usepackage[hidelinks]{hyperref}
\newcommand{\E}{\mathbb E}
\newcommand{\Pp}{\mathbb P}
\newcommand{\R}{\mathbb R}
\newcommand{\Var}{\operatorname{Var}}
\newcommand{\Cov}{\operatorname{Cov}}
\newcommand{\tr}{\operatorname{tr}}
\newcommand{\Diag}{\operatorname{Diag}}
\newcommand{\op}{\operatorname{op}}

\newcommand{\Kcal}{\mathcal K}
\newcommand{\norm}[1]{\left\lVert#1\right\rVert}
\newcommand{\spSign}{\mathcal{S}}
\newcommand{\bx}{\mathbf{x}}
\newcommand{\by}{\mathbf{y}}
\newcommand{\bz}{\mathbf{z}}
\newcommand{\bg}{\mathbf{g}}
\newcommand{\bs}{\mathbf{s}}
\newcommand{\bu}{\mathbf{u}}
\newcommand{\bv}{\mathbf{v}}
\newcommand{\bw}{\mathbf{w}}
\newcommand{\ba}{\mathbf{a}}
\newcommand{\bb}{\mathbf{b}}
\newcommand{\bl}{\mathbf{l}}
\newcommand{\bq}{\mathbf{q}}
\newcommand{\br}{\mathbf{r}}
\newcommand{\bt}{\mathbf{t}}
\newcommand{\bh}{\mathbf{h}}
\newcommand{\bmu}{\boldsymbol{\mu}}
\newcommand{\bdelta}{\boldsymbol{\delta}}
\newcommand{\bDelta}{\boldsymbol{\Delta}}
\newcommand{\bPi}{\boldsymbol{\Pi}}
\newcommand{\bSigma}{\boldsymbol{\Sigma}}
\newcommand{\bOmega}{\boldsymbol{\Omega}}
\newcommand{\bXi}{\boldsymbol{\Xi}}
\newcommand{\bA}{\mathbf{A}}
\newcommand{\bB}{\mathbf{B}}
\newcommand{\bK}{\mathbf{K}}
\newcommand{\bL}{\mathbf{L}}
\newcommand{\bD}{\mathbf{D}}
\newcommand{\bE}{\mathbf{E}}
\newcommand{\bQ}{\mathbf{Q}}
\newcommand{\bGamma}{\boldsymbol{\Gamma}}
\newcommand{\bLambda}{\boldsymbol{\Lambda}}
\newcommand{\bP}{\mathbf{P}}
\newcommand{\bS}{\mathbf{S}}
\newcommand{\bM}{\mathbf{M}}
\newcommand{\bN}{\mathbf{N}}
\newcommand{\bC}{\mathbf{C}}
\newcommand{\bH}{\mathbf{H}}
\newcommand{\bR}{\mathbf{R}}
\newcommand{\bI}{\mathbf{I}}
\newtheorem{theorem}{Theorem}
\newtheorem{assumption}{Assumption}
\newtheorem{remark}{Remark}
\newtheorem{lemma}{Lemma}
\newtheorem{proposition}{Proposition}
\newcommand{\ind}{\mathbb{I}}
\newcommand{\bG}{\mathbf{G}}
\newcommand{\bPhi}{\boldsymbol{\Phi}}
\newcommand{\bPsi}{\boldsymbol{\Psi}}
\newcommand{\papertitle}{Two-sample tests for principal eigenvalues and eigenvectors in high-dimensional elliptical factor models}

\begin{document}

\title{\bfseries\papertitle}

\author{Xinyue Xu$^1$, Mengtao Wen$^1$, and Long Feng$^1$\\
School of Statistics and Data Science,\\
LEBPS, KLMDASR and LPMC, Nankai University$^1$}

\date{}
\maketitle

\begin{abstract}
Changes in principal eigenvalues and eigendirections provide complementary
diagnostics of structural instability in factor models, but covariance-based
inference can be unreliable under heavy tails. We develop two-sample tests
for equality of these features of trace-normalized shape matrices under
high-dimensional elliptical factor models. Our approach combines Tyler's
identity with orthogonalization and cross-fitting to accommodate general
location and precision pilots. We derive a spectral limit theory that yields
asymptotically valid calibration without upper-tail moment assumptions on the
radial variables, allowing the two populations to have different radial
distributions and factor ranks. Simulations demonstrate good size control
and gains in size-adjusted power over covariance-based benchmarks under
heavy-tailed elliptical distributions. An application to S\&P~500 stock
returns illustrates how the tests distinguish changes in relative component
strength from changes in component orientation.
\end{abstract}

\section{Introduction}

\subsection{Background and motivation}

Spiked covariance models provide a framework for studying high-dimensional
data in which a few dominant directions capture systematic variation
against a background of weaker components \citep{Johnstone2001}.
An important instance is the approximate factor model with pervasive
factors, which summarizes dependence among many variables through a small
number of common factors and has broad applications in economics and
finance, including asset pricing and macroeconomic forecasting
\citep{FamaFrench1993,StockWatson2002}.
When the idiosyncratic covariance has bounded eigenvalues, pervasiveness
yields diverging leading eigenvalues, with the corresponding eigenspace
asymptotically recovering the factor loading space \citep{fan2013large}.
Principal component analysis (PCA) therefore plays a central role in the
analysis of large factor models: the empirical spectrum informs the
selection of the number of factors, while the leading eigenvectors and
component scores provide estimates of factor loadings and latent factors
\citep{BaiNg2002,Bai2003}.

However, the standard PCA-based analysis of static factor models typically
assumes that factor loadings remain constant over the estimation period.
In economic and financial applications, changes in policy regimes, market
conditions, and institutional arrangements can alter the relationships
between observed variables and common factors. Ignoring substantial
changes in factor loadings can compromise the interpretation of estimated
factors and inflate the number of factors selected from the data
\citep{BreitungEickmeier2011}. These concerns make tests of loading stability
an important diagnostic tool for empirical factor analysis and have
motivated a substantial body of research on structural breaks and parameter
instability in large factor models
\citep{ChenDoladoGonzalo2014,HanInoue2015}.

An omnibus rejection of loading stability does not, by itself, reveal
which features of the underlying factor structure have changed.
For example, consider a stock-return panel whose leading principal
component captures broad market movements. An increase in the leading
eigenvalue with an unchanged eigenvector indicates stronger fluctuations
along the same pattern of market-wide co-movement. By contrast, a change
in the eigenvector---for instance, a shift in weight from financial stocks
toward technology stocks---indicates a change in the composition of the
common component, even if its variance remains unchanged.
Testing individual principal eigenvalues and eigendirections across
populations or periods can thus complement overall loading-stability tests
with a more informative diagnosis of structural change
\citep{FanLiXiaZheng2026}.

Building on this perspective, \citet{FanLiXiaZheng2026} establish
two-sample tests for eigenvalues, eigenvalue proportions, and eigendirections under divergent spiked covariance models, where they assume that the population principal component scores, standardized to unit variance, are mutually independent and have uniformly
bounded fourth moments. These assumptions permit non-Gaussian observations, but can be
restrictive for financial data. In particular, a common random volatility
scale can make these scores dependent even when they are uncorrelated,
while sufficiently heavy tails can violate the fourth-moment condition.
Heavy tails and nonlinear dependence are well documented in financial
returns \citep{Cont2001}. These features motivate diagnostic methods
that accommodate them.

Elliptical distributions provide a natural setting for accommodating
heavy tails and dependence among principal component scores
\citep{FanLiuWang2018}. Their radial-angular representation separates
radial variation from a shape matrix that remains meaningful even when
the covariance does not exist. Since this matrix is identifiable up
to a positive scalar, we fix its scale by setting its trace equal to the
dimension. Therefore, we seek a two-sample diagnostic framework for the
principal eigenvalues and eigendirections of this shape matrix, without
imposing upper-tail moment conditions on the radial variable.
Developing such a framework entails two main challenges. 
First, consistency of robust shape estimators does not imply the first-order spectral expansions needed for testing; valid calibration requires identifying the leading fluctuations and controlling the remainder at the inferential scale. 
Second, small estimation errors along the many nonspiked directions can accumulate and shift the null distribution of eigendirection statistics, requiring appropriate centering for valid calibration.

Robust methods for elliptical data have received considerable attention.
\citet{Tyler1987} developed a scatter estimator whose
fixed-dimensional limiting distribution is independent of the radial law.
In high-dimensional settings, \citet{HanLiu2018} proposed elliptical component
analysis based on multivariate Kendall's tau, obtaining principal
eigenspace estimation guarantees without radial moment assumptions.
For elliptical factor models, \citet{FanLiuWang2018} established a robust
Principal Orthogonal complEment Thresholding (POET) framework for
covariance estimation under fourth-moment conditions,
while \citet{XuMaWangFeng2025} developed spatial-sign and Tyler-based extensions for scatter and precision estimation. 
These developments provide useful estimation tools for our problem, but further
theory is needed to characterize the relevant spectral fluctuations for the current tasks.

\subsection{Our contributions}

In this paper, we develop a two-sample diagnostic framework for the principal eigenvalues and eigendirections of trace-normalized shape matrices under high-dimensional elliptical factor models. Our main contributions are as follows.

\begin{enumerate}[label=(\roman*),leftmargin=*]
\item \textbf{Orthogonal spectral inference with general pilots.}
Using Tyler's identity and trace normalization, we construct a shape estimator from general location and precision pilots. 
Neyman orthogonalization removes the first-order effect of precision-pilot error, while elliptical symmetry already ensures first-order insensitivity to location perturbations.
We then use cross-fitting to separate nuisance estimation from evaluation.
We derive spectral expansions whose leading terms depend only on the angular component, yielding asymptotic normality for spiked eigenvalues and a limit theory for eigendirections.
Under the stated regularity conditions, \(o_{\Pp}(n^{-1/4})\) rates for the precision error in operator norm, up to a common positive scale, and the standardized location error suffice without upper-tail radial moment assumptions.

\item \textbf{Two-sample diagnostics.}
We derive studentized normal tests for spiked eigenvalues and tests for equality of the corresponding eigendirections.
For eigendirection testing, centering removes the cumulative contributions from non-spiked directions, and a fitted Gaussian quadratic form calibrates the remaining fluctuations.
The resulting critical values and \(p\)-values are asymptotically valid even when the two samples have different radial distributions.
Simulations assess finite-sample size and power under Gaussian and heavy-tailed elliptical distributions.
An application to S\&P~500 stock returns illustrates how the tests distinguish changes in relative component strength from changes in the weights on individual stocks.
\end{enumerate}

\subsection{Organization and notation}

The remainder of the paper is organized as follows.
Section~\ref{sec:model-and-targets} introduces the elliptical factor model and formulates the testing problems.
In Section~\ref{sec:single-sample}, we construct the shape estimator and establish its one-sample spectral properties.
The two-sample tests for principal eigenvalues and eigendirections are presented in Section~\ref{sec:two-sample-tests}.
Section~\ref{sec:simulation} presents the numerical studies, including simulations and an application to financial returns in Section~\ref{sec:application}.
We conclude with a discussion in Section~\ref{sec:conclusion}.
Proofs and additional simulation and empirical results are provided in the Supplementary Material.

We use the following notation throughout.
For positive deterministic sequences \(a_n\) and \(b_n\), we write \(a_n\asymp b_n\) if their ratio is bounded above and away from zero, and use \(O(\cdot)\) and \(o(\cdot)\) for deterministic orders.
For a real number \(x\), let \(x_+=\max(x,0)\).
Vectors are denoted by bold lowercase letters, with \(\langle\bu,\bv\rangle=\bu^\top\bv\) and \(\|\bu\|_2\) denoting the Euclidean inner product and norm.
For \(0<v<1\), \(\|\bu\|_v=(\sum_p|u_p|^v)^{1/v}\) is the \(\ell_v\) quasi-norm, and \(\|\bu\|_0\) counts the nonzero entries of \(\bu\).
Matrices are denoted by bold uppercase letters; \(\bA^\top\), \(\tr(\bA)\), \(\|\bA\|_{\op}\), and \(\|\bA\|_{\mathrm F}\) denote the transpose, trace, operator norm, and Frobenius norm, respectively.
We write \((\bA)_{ps}\) and \((\bA)_{p\cdot}\) for the \((p,s)\)th entry and the \(p\)th row of \(\bA\), respectively, omitting the parentheses when no ambiguity arises.
We write \(\bI_d\) for the \(d\times d\) identity matrix and \(\Diag(\cdot)\) for the diagonal or block diagonal matrix formed from the specified entries or blocks.
For a symmetric matrix \(\bA\), \(\lambda_{\min}(\bA)\) and \(\lambda_{\max}(\bA)\) denote its smallest and largest eigenvalues, while \(\bA\succeq0\) and \(\bA\succ0\) indicate positive semidefiniteness and positive definiteness.
For random quantities, \(\Pp\), \(\E\), \(\Var\), and \(\Cov\) denote probability, expectation, variance, and covariance, and \(\ind(A)\) is the indicator of an event \(A\).
Finally, \(O_{\Pp}(\cdot)\) and \(o_{\Pp}(\cdot)\) denote stochastic orders, and \(\xrightarrow{\Pp}\), \(\xrightarrow{\mathcal D}\), and \(\stackrel{\mathcal{D}}{=}\) denote convergence in probability, convergence in distribution, and equality in distribution, respectively.


\section{Model and Testing Problems}
\label{sec:model-and-targets}

\subsection{Elliptical factor model}

We consider a \(d\)-dimensional elliptical random vector \(\bx\) with the
stochastic representation
\begin{equation}\label{eq:model}
 \bx\stackrel{\mathcal{D}}{=}\bmu+\xi\bSigma^{1/2}\bs,
\end{equation}
where \(\bmu\in\R^d\) is a location parameter, \(\bSigma\) is a positive
definite scatter matrix, the angular vector \(\bs\) is uniformly distributed on the unit
sphere \(\mathbb S^{d-1}=\{\bu\in\R^d:\norm{\bu}_2=1\}\), and the radial variable \(\xi>0\) is independent of \(\bs\).
Here, the scatter matrix is identifiable up to a positive scalar, so we fix its scale by imposing \(\tr(\bSigma)=d\) and henceforth refer to \(\bSigma\) as the shape matrix.


To describe the principal components of \(\bSigma\), we assume that it admits the approximate factor structure \citep{fan2013large}
\begin{equation}\label{eq:factor-structure}
 \bSigma=\bB\bB^\top+\bSigma_{\varepsilon},
 \qquad \bB\in\R^{d\times m},\qquad \bSigma_{\varepsilon}\succ0,
\end{equation}
where \(\bB\) is the factor loading matrix, \(m\) is the number of common
factors, and \(\bSigma_{\varepsilon}\) represents the idiosyncratic component. This
structure is imposed directly on the shape matrix and does not require
\(\bx\) to have a finite covariance matrix. We use the following
regularity conditions.

\begin{assumption}[Elliptical factor model]\label{ass:factor-structure}
The random vector \(\bx\) follows the elliptical model in \eqref{eq:model}, with \(\tr(\bSigma)=d\). Its shape matrix admits the factor decomposition in \eqref{eq:factor-structure}, where the number of factors \(m\geq1\) remains fixed as \(d\to\infty\). The following conditions hold:
\begin{enumerate}[label=(\roman*),leftmargin=*]
 \item There exists a positive definite matrix \(\bXi\in\R^{m\times m}\)
 with distinct eigenvalues \(\kappa_1>\cdots>\kappa_m>0\) such that
 \(\|d^{-1}\bB^\top\bB-\bXi\|_{\op}\longrightarrow0\).
 \item The rows of the loading matrix are uniformly bounded, i.e., 
 \(\max_{1\le p\le d}\norm{\bB_{p\cdot}}_2\le C_B\)
 for some constant \(C_B<\infty\).
 \item There exist constants \(0<c_{\varepsilon}<C_{\varepsilon}<\infty\) such that \(c_{\varepsilon}\le\lambda_{\min}(\bSigma_{\varepsilon})
  \le\lambda_{\max}(\bSigma_{\varepsilon})\le C_{\varepsilon}\).
\end{enumerate}
\end{assumption}

The convergence in Assumption~\ref{ass:factor-structure}(i) is the pervasiveness condition used in
approximate factor models \citep{fan2013large}. Together with the bounded
loadings in Assumption~\ref{ass:factor-structure}(ii), it ensures that the common factors affect a
nonvanishing fraction of the observed variables. Assumption~\ref{ass:factor-structure}(iii) allows dependence
among the idiosyncratic components while keeping the eigenvalues of
\(\bSigma_{\varepsilon}\) bounded above and away from zero.

Let \(\lambda_1\ge\cdots\ge\lambda_d>0\) denote the eigenvalues of
\(\bSigma\), with corresponding orthonormal eigenvectors
\(\bu_1,\ldots,\bu_d\). By Weyl's inequality,
Assumption~\ref{ass:factor-structure} implies
\begin{equation}
 {\lambda_j}/{d}\longrightarrow\kappa_j \quad \text{for}\quad 1\le j\le m,
 \qquad\text{and}\qquad 
 c_{\varepsilon}\le\lambda_j\le C_{\varepsilon} \quad \text{for}\quad m<j\le d.
\end{equation}
Thus, the first \(m\) eigenvalues are of order \(d\), whereas the remaining eigenvalues stay bounded. 
To ensure that each principal eigenvector is uniquely defined up to sign for all sufficiently large \(d\), Assumption~\ref{ass:factor-structure}(i) requires the limits \(\kappa_1,\ldots,\kappa_m\) to be distinct.
This condition separates each principal eigenvalue from the rest of the spectrum by a gap of order \(d\).
Our inferential targets are these principal eigenvalues and their associated
eigendirections.

\subsection{Two-sample hypotheses}
\label{subsec:testing-problems}

Let \(\{\bx_i^{(r)}\}_{i=1}^{n_r}\) for \(r=1,2\)  be two independent
samples, each consisting of independent and identically distributed (i.i.d.)
observations from a population satisfying
Assumption~\ref{ass:factor-structure} with its own fixed factor rank
\(m_r\). Denote their shape matrices by
\(\bSigma^{(1)}\) and \(\bSigma^{(2)}\), with
\(\tr(\bSigma^{(r)})=d\), and write
\(\lambda_j^{(r)}\) and \(\bu_j^{(r)}\) for the corresponding ordered eigenvalues and unit eigenvectors. 
To accommodate heterogeneous populations, we do not require their locations, radial distributions, or factor ranks to coincide.
For a fixed \(j\in\{1,\ldots,m\}\), where \(m=\min(m_1,m_2)\), we consider the following two tests concerning the \(j\)th principal component in each population.

The first concerns equality of the \(j\)th principal eigenvalues:
\begin{equation}\label{eq:null-eigenvalue}
 \mathbb{H}_{0,\lambda,j}:\ \lambda_j^{(1)}=\lambda_j^{(2)}
 \qquad\text{versus}\qquad
 \mathbb{H}_{1,\lambda,j}:\ \lambda_j^{(1)}\ne\lambda_j^{(2)}.
\end{equation}
It is worth noting that testing equality of eigenvalues here is equivalent to testing
equality of their proportions of the total trace, 
since both shape matrices have trace \(d\).
The second concerns equality of the \(j\)th principal eigendirections.
Because eigenvectors are defined only up to sign, we formulate the
hypotheses in terms of the absolute inner product of the unit eigenvectors:
\begin{equation}\label{eq:null-eigenvector}
 \mathbb{H}_{0,u,j}:\ \left|\langle\bu_j^{(1)},\bu_j^{(2)}\rangle\right|=1
 \qquad\text{versus}\qquad
 \mathbb{H}_{1,u,j}:\ \left|\langle\bu_j^{(1)},\bu_j^{(2)}\rangle\right|<1.
\end{equation}
Equivalently, the null hypothesis states that
\(\bu_j^{(1)}=\pm\bu_j^{(2)}\). 
These two testing problems assess the
relative strength and orientation of a principal component, respectively.


\section{Shape Estimation and Spectral Asymptotics}
\label{sec:single-sample}

To construct the two-sample tests, we first develop one-sample
estimation and asymptotic theory. The resulting procedures for two-sample testing will be
presented in Section~\ref{sec:two-sample-tests}. Throughout this section,
we consider independent and identically distributed observations
\(\bx_1,\ldots,\bx_n\) from a population satisfying
Assumption~\ref{ass:factor-structure} and suppress the group index. Here,
\(n\) and \(m\) denote the sample size and factor rank of this population,
respectively.

\subsection{A pilot-based Tyler update}

Tyler's M-estimator provides a natural starting point for estimating
\(\bSigma\) under elliptical symmetry \citep{Tyler1987}. Its population
counterpart is characterized by Tyler's identity: writing
\(\bOmega=\bSigma^{-1}\), we have
\begin{equation}\label{eq:population-tyler-identity}
 \E\left[
 d\,\frac{(\bx-\bmu)(\bx-\bmu)^\top}
 {(\bx-\bmu)^\top\bOmega(\bx-\bmu)}
 \right]=\bSigma.
\end{equation}
We refer to the matrix inside the square brackets in
\eqref{eq:population-tyler-identity} as the Tyler score.
Substituting the representation in \eqref{eq:model} cancels the radial
factor \(\xi^2\) between the numerator and denominator, reducing this
score to \(d\bSigma^{1/2}\bs\bs^\top\bSigma^{1/2}\), whose expectation
is \(\bSigma\). This identity does not require any moment condition on
the radial variable.

In the classical setting with fixed \(d\), replacing the expectation by
a sample average and the unknown precision by the inverse of the scatter
estimate leads to a fixed-point equation. Given a location estimate
\(\widehat{\bmu}\), Tyler's M-estimator is a positive definite solution
of
\begin{equation}\label{eq:classical-tme}
 \widehat{\bSigma}_{T}
 =\frac{d}{n}\sum_{i=1}^n
 \frac{(\bx_i-\widehat{\bmu})(\bx_i-\widehat{\bmu})^\top}
 {(\bx_i-\widehat{\bmu})^\top
  \widehat{\bSigma}_{T}^{-1}(\bx_i-\widehat{\bmu})},
 \qquad \tr(\widehat{\bSigma}_{T})=d.
\end{equation}
Under the usual subspace conditions on the centered observations,
\eqref{eq:classical-tme} admits a unique solution after fixing its scale
\citep{Tyler1987}. When \(d>n\), however, its right-hand side has rank
at most \(n\) and cannot equal a positive definite matrix.
Regularized Tyler estimators address this
problem by introducing shrinkage toward a positive definite target
\citep{sun2014regularized}. 
Nevertheless, a nonvanishing shrinkage penalty generally changes the population fixed point and introduces bias relative to the original shape matrix \citep{KammounEtAl2016}. 
For inference on the eigenstructure of
\(\bSigma\), the resulting bias has to be controlled at the
relevant inferential scale.

This motivates using \eqref{eq:population-tyler-identity} directly to
construct a sample update. Let \(\widehat{\bmu}\) and
\(\widehat{\bP}\succ0\) be pilot estimators of the nuisance parameters
\(\bmu\) and \(\bOmega\), respectively. Evaluating the Tyler score at
these pilots and averaging over the sample yields the update
\begin{equation}\label{eq:pilot-tyler-update}
 \widehat{\bSigma}_{\mathrm{plug}}
 =\frac{d}{n}\sum_{i=1}^n
 \frac{(\bx_i-\widehat{\bmu})(\bx_i-\widehat{\bmu})^\top}
 {(\bx_i-\widehat{\bmu})^\top
  \widehat{\bP}(\bx_i-\widehat{\bmu})}.
\end{equation}
Here, to ensure that the update is well defined, a summand is set to zero
when \(\bx_i=\widehat{\bmu}\); the same convention applies to analogous
expressions below. With the oracle nuisance values \((\bmu,\bOmega)\),
this update is an unbiased estimator of \(\bSigma\); although the update
is singular when \(d>n\), its leading eigenpairs can still be studied.
When the nuisance parameters are estimated, suitable accuracy conditions
on both pilots ensure that their errors do not affect the first-order
spectral limits after the refinement in the next subsection. These
conditions accommodate flexible pilot constructions; one implementation
is given in Section~\ref{subsec:pilot-construction}.


\subsection{Neyman orthogonalization and cross-fitting}
\label{subsec:orth-estimator}

The plug-in update in \eqref{eq:pilot-tyler-update} involves estimation of both location and precision. 
Under elliptical symmetry and the regularity conditions stated in Section~\ref{subsec:single-sample-asymptotics}, the expected Tyler score has zero derivative with respect to location at the true nuisance values, whereas its derivative with respect to the precision matrix is generally nonzero.
We therefore apply Neyman orthogonalization to
eliminate the first-order sensitivity to the precision pilot.
Furthermore, we use cross-fitting for both nuisance
pilots to reduce the need for restrictive complexity conditions \citep{ChernozhukovEtAl2018}. 
For simplicity, we directly give the following construction combining these two ingredients.

Specifically, we randomly partition \(\{1,\ldots,n\}\) into a fixed number \(K\ge2\) of disjoint
folds \(I_1,\ldots,I_K\) with
\(N_\ell=|I_\ell|\asymp n\). For each \(\ell\in\{1,\ldots,K\}\), we construct pilots
\(\widehat{\bmu}_{-\ell}\) and \(\widehat{\bP}_{-\ell}\succ0\) for
\(\bmu\) and \(\bOmega\), respectively, using only observations outside
\(I_\ell\).
Then, we summarize the Tyler scores evaluated at these held-out pilots over the corresponding fold, and average the precision pilots over all folds:
\begin{equation}\label{eq:cross-fitted-tyler}
 \begin{aligned}
 \widehat{\bSigma}_{\mathrm{cf}}
 &=\frac{d}{n}\sum_{\ell=1}^K\sum_{i\in I_\ell}
 \frac{
 (\bx_i-\widehat{\bmu}_{-\ell})
 (\bx_i-\widehat{\bmu}_{-\ell})^\top
 }{
 (\bx_i-\widehat{\bmu}_{-\ell})^\top
 \widehat{\bP}_{-\ell}
 (\bx_i-\widehat{\bmu}_{-\ell})
 },
 \qquad \bar{\bP}
 =\sum_{\ell=1}^K\frac{N_\ell}{n}\widehat{\bP}_{-\ell}.
 \end{aligned}
\end{equation}
To impose the trace constraint, we normalize the cross-fitted update and calibrate the averaged precision pilot:
\[
 \widehat{c}_s=\frac{\tr(\widehat{\bSigma}_{\mathrm{cf}})}{d},\qquad
 \widetilde{\bSigma}=\frac{\widehat{\bSigma}_{\mathrm{cf}}}{\widehat{c}_s},\qquad
 \widetilde{\bOmega}=\widehat{c}_s\,\bar{\bP}.
\]
Finally, we apply the orthogonal correction
to this trace-normalized update to get the shape estimator
\begin{equation}\label{eq:pre-poet-orthogonalization}
 \widehat{\bSigma}
 =\widetilde{\bSigma}+{2}{d}^{-1}\left[
 \widetilde{\bSigma}\widetilde{\bOmega}\widetilde{\bSigma}
 -d^{-1}{\tr(\widetilde{\bSigma}^2\widetilde{\bOmega})}\widetilde{\bSigma}
 \right].
\end{equation}
We have \(\tr(\widetilde{\bSigma})=d\) due to the trace normalization, and the correction term in
\eqref{eq:pre-poet-orthogonalization} is trace-free, so that \(\tr(\widehat{\bSigma})=d\) as well. 
The spectral decomposition of \(\widehat{\bSigma}\) is
\[
 \widehat{\bSigma}
 =\sum_{j=1}^d\widehat\lambda_j\widehat{\bu}_j\widehat{\bu}_j^\top,
 \qquad \widehat\lambda_1\ge\cdots\ge\widehat\lambda_d,
\]
where \(\widehat{\bu}_1,\ldots,\widehat{\bu}_d\) are the corresponding
orthonormal eigenvectors.


\begin{remark}[Precision pilots up to scale]\label{rem:pilot-scale}
The construction also allows precision pilots that approximate
\(\bP_*=c_d^{-1}\bOmega\), where \(c_d>0\) is an unknown
deterministic scale factor that may depend on \(d\). At the oracle inputs
\((\bmu,\bP_*)\), the expected Tyler score equals
\(c_d\bSigma\). Trace normalization recovers \(\bSigma\),
while multiplying \(\bP_*\) by the retained trace scale \(c_d\)
recovers \(\bOmega\). The empirical factor \(\widehat{c}_s\) above provides
the corresponding sample calibration. Therefore, trace normalization together
with precision calibration removes the common scalar ambiguity without
requiring knowledge of \(c_d\).
\end{remark}

\subsection{Asymptotic expansions and limiting distributions}
\label{subsec:single-sample-asymptotics}

We first establish an asymptotic expansion for
\(\widehat{\bSigma}\) and then obtain the
eigenvalue and eigendirection results through spectral perturbation.
Throughout this subsection, we assume that \(n,d\to\infty\). We impose separate
conditions on the radial distribution and the nuisance estimates,
without prescribing a particular pilot construction. 
Let \(\bar\xi=\xi/\sqrt d\) and
\(\zeta_1=\E\{\|\bx-\bmu\|_2^{-1}\}\).

\begin{assumption}[Inverse-radial moments]\label{ass:inverse-radial-moments}
 There exist constants
 \(c_\xi\in(0,1)\) and \(K_\xi<\infty\)
 such that
 \(\E(\bar\xi^{-q})<\infty\) and 
  \({\{\E(\bar\xi^{-q})\}^{1/q}}/{\E(\bar\xi^{-1})}
  \le K_\xi\sqrt q\) for every positive integer \(q\le\lfloor c_\xi d\rfloor\).
\end{assumption}

\begin{assumption}[Nuisance estimation]\label{ass:nuisance-estimation}
 The nuisance estimates are constructed on the training samples specified in Section~\ref{subsec:orth-estimator}, and
 \(\widehat{\bP}_{-\ell}\succ0\). For a deterministic
 \(c_d>0\) that is the same across all folds and a deterministic sequence
 \(a_n=o(n^{-1/4})\),
 \begin{equation}\label{eq:general-nuisance-rates}
  \max_{\ell\le K}
  \zeta_1\|\widehat{\bmu}_{-\ell}-\bmu\|_2
  =O_{\Pp}(a_n),\qquad 
  \max_{\ell\le K}
  \|c_d\widehat{\bP}_{-\ell}-\bOmega\|_{\op}
  =O_{\Pp}(a_n).
 \end{equation}
\end{assumption}

Assumption~\ref{ass:inverse-radial-moments} bounds the growth of inverse-radial moments to control the effect of location estimation on observations close to the center.
Related conditions on inverse-distance moments also appear in the literature on high-dimensional spatial-sign inference \citep{ZouPengFengWang2014,WangPengLi2015}.
As explained in Remark~\ref{rem:pilot-scale}, the precision pilot need only estimate \(\bOmega\) up to scale, so \eqref{eq:general-nuisance-rates} allows the scalar \(c_d\).
Assumption~\ref{ass:nuisance-estimation} requires only \(o_{\Pp}(n^{-1/4})\) accuracy for the standardized location error and the precision error in operator norm.
Verification of Assumption~\ref{ass:nuisance-estimation} for a concrete implementation is
deferred to Section~\ref{subsec:pilot-construction}.

We now identify the first-order term of the shape-matrix estimator.
For each \(j\le m\), we define the matrix
\(\bG_j=\sum_{k\ne j}{\bu_k\bu_k^\top}/{(\lambda_j-\lambda_k)}\).
We also write \(\epsilon_n=a_n+n^{-1/2}\) and
\(\rho_n=\epsilon_n^2\), so \(\sqrt n\rho_n\to0\). 

\begin{theorem}[Expansion of the shape estimator]\label{thm:common-expansion}
Under Assumptions~\ref{ass:factor-structure}--\ref{ass:nuisance-estimation},
\begin{equation}\label{eq:common-matrix-expansion}
 \widehat{\bSigma}-\bSigma
 =\frac{d+2}{d}\frac1n\sum_{i=1}^n\bPsi(\bx_i)+\bR_n,
\end{equation}
where \(\bPsi(\bx)
 =d\bSigma^{1/2}\bs\bs^\top\bSigma^{1/2}
  -(\bs^\top\bSigma\bs)\bSigma\), and \(\bR_n\) satisfies
\begin{equation}\label{eq:common-spectral-remainder}
 \lambda_j^{-1}|\bu_j^\top\bR_n\bu_j|
 +\|\bG_j\bR_n\bu_j\|_2=O_{\Pp}(\rho_n)\qquad \text{for each}\quad  j\le m.
\end{equation}
In addition,
\begin{equation}\label{eq:common-matrix-rate}
 d^{-1}\|\widehat{\bSigma}-\bSigma\|_{\mathrm F}
 =O_{\Pp}(\epsilon_n).
\end{equation}
\end{theorem}

Theorem~\ref{thm:common-expansion} expresses the shape estimation error as an average of mean-zero matrix-valued scores plus a remainder term.
The score \(\bPsi(\bx)\) depends only on the angular component \(\bs\), so the radial variable \(\xi\) does not enter the leading fluctuation.
The bounds in \eqref{eq:common-spectral-remainder} make the remainder negligible at the root-\(n\) scale for both relative eigenvalue errors and eigendirection errors.

\begin{remark}[The role of orthogonalization]
\label{rem:neyman-orthogonality}
Before orthogonalization, the cross-fitted and trace-normalized estimator has the expansion
\[
 \begin{aligned}
 \widetilde{\bSigma}-\bSigma
 ={}&\frac1n\sum_{i=1}^n\bPsi(\bx_i) -\frac{2}{d+2}\left[
 \bSigma\bar{\bDelta}_P\bSigma
 -\frac{\tr(\bSigma^2\bar{\bDelta}_P)}{d}\bSigma
 \right]+\bR_{N,n},
 \end{aligned}
\]
where \(\bar{\bDelta}_P=c_d\bar{\bP}-\bOmega\) and
\(d^{-1}\|\bR_{N,n}\|_{\mathrm F}=O_{\Pp}(\rho_n)\).
The second term is linear in the precision-pilot error and can affect
relative eigenvalue errors and eigendirection errors.
In contrast, the correction in \eqref{eq:pre-poet-orthogonalization} cancels it, permitting \(a_n=o(n^{-1/4})\) in Assumption~\ref{ass:nuisance-estimation}. 
We orthogonalize only with respect to the precision pilot, since the population linear effect of location estimation already vanishes by elliptical symmetry.
\end{remark}

For the spiked eigenvalues, applying eigenvalue perturbation theory \citep{GreenbaumLiOverton2020} to the results of Theorem~\ref{thm:common-expansion} yields the following relative-error expansion and asymptotic normality.

\begin{theorem}[One-sample spiked eigenvalues]
\label{thm:one-sample-eigenvalue}
Under Assumptions~\ref{ass:factor-structure}--\ref{ass:nuisance-estimation},
for each fixed \(j\le m\),
\begin{equation}\label{eq:eigenvalue-expansion-main}
 \frac{\widehat\lambda_j-\lambda_j}{\lambda_j}
 =\frac{d+2}{d}\frac1n\sum_{i=1}^n
   \varphi_{\lambda,j}(\bx_i)+O_{\Pp}(\rho_n),
\end{equation}
where \(\varphi_{\lambda,j}(\bx)
 = \lambda_j^{-1} \bu_j^\top\bPsi(\bx)\bu_j = d(\bu_j^\top \bs)^2 - \bs^\top\bSigma\bs\) is the score function.
Moreover,
\begin{equation}\label{eq:eigenvalue-clt-main}
 \frac{\sqrt n(\widehat\lambda_j-\lambda_j)}
 {\lambda_j\sigma_{\lambda,j}}
 \xrightarrow{\mathcal D}N(0,1),
\end{equation}
where \(\sigma_{\lambda,j}^2
 ={2(d+2)}d^{-3}
  \{d^2-2d\lambda_j+\sum_{k=1}^d\lambda_k^2\}\).
\end{theorem}

In Theorem~\ref{thm:one-sample-eigenvalue}, the inferential target of our eigenvalue estimator \(\widehat{\lambda}_j\) remains \(\lambda_j\), and the division by \(\lambda_j\) expresses the fluctuation on its natural
scale, since the spiked eigenvalues diverge with \(d\).
The variance \(\sigma_{\lambda, j}^2\) includes the full spectrum \(\sum_{k =1}^d \lambda_k^2\), because trace normalization
couples the target eigenvalue to the total shape. 

Finally, we turn to the spiked eigendirections. 
Without loss of generality, we align the sign of \(\widehat{\bu}_j\) such that \(\bu_j^\top\widehat{\bu}_j\ge0\).
For \(k\ne j\), define
\begin{equation}\label{eq:eigenvector-influence-score}
 \eta_{jk}(\bx)
 =\frac{\bu_k^\top\bPsi(\bx)\bu_j}{\lambda_j-\lambda_k}
 =\frac{d\sqrt{\lambda_j\lambda_k}}{\lambda_j-\lambda_k}
   (\bu_j^\top\bs)(\bu_k^\top\bs),
\end{equation}
and its scaled variance as 
\begin{equation}\label{eq:direction-population-weight}
 \sigma_{u,jk}^2
 :=\Var\left\{\frac{d+2}{d}\eta_{jk}(\bx)\right\}
 =\frac{(d+2)\lambda_j\lambda_k}
 {d(\lambda_j-\lambda_k)^2}.
\end{equation}
For \(k\le m\) and \(k\ne j\), \(\sigma_{u,jk}^2\) converges to \(\sigma_{u,jk}^{*2}:=\kappa_j\kappa_k/(\kappa_j-\kappa_k)^2\).

\begin{theorem}[One-sample spiked eigendirections]
\label{thm:one-sample-eigenvector}
Under Assumptions~\ref{ass:factor-structure}--\ref{ass:nuisance-estimation},
for each fixed \(j\le m\),
\begin{equation}\label{eq:one-sample-eigenvector-expansion}
 \widehat{\bu}_j-\bu_j
 =\frac{d+2}{d}\frac1n\sum_{i=1}^n
   \sum_{k\ne j}\eta_{jk}(\bx_i)\bu_k
  +\br_{n,j},
 \qquad
 \|\br_{n,j}\|_2=O_{\Pp}(\rho_n).
\end{equation}
Moreover, the squared sine of the angle between \(\widehat{\bu}_j\) and \(\bu_j\) satisfies
\begin{equation}\label{eq:one-sample-full-eigenvector-limit}
 n\left[1-\left|\bu_j^\top\widehat{\bu}_j\right|^2\right]-b_{u,j}
 \xrightarrow{\mathcal D}
 \sum_{\substack{1\le k\le m\\k\ne j}}
 \sigma_{u,jk}^{*2}Z_k^2,
\end{equation}
where \(b_{u,j}=\sum_{k=m+1}^d\sigma_{u,jk}^2\), and
\(\{Z_k\}_{1\le k\le m,\,k\ne j}\) are independent standard normal variables.
When \(m=1\), the quadratic form and the sum are interpreted as zero,
and the centered quantity converges to zero in probability.
\end{theorem}

In Theorem~\ref{thm:one-sample-eigenvector}, \(b_{u,j}\) accounts for estimation error along the nonspiked eigenvectors.
Subtracting \(b_{u,j}\) in \eqref{eq:one-sample-full-eigenvector-limit} leaves the weighted chi-square fluctuation generated by the finitely many other spiked directions.
For one-sample eigendirection inference, we require estimates of both the centering term \(b_{u,j}\) and the weights \(\sigma_{u,jk}^{*2}\).
The trace constraint on \(\bSigma\) gives \(b_{u,j} =(d+2)d^{-1}\lambda_j^{-1} (d-\sum_{k=1}^m\lambda_k) +O(d^{-1})\), so both terms can be estimated from the spiked eigenvalues. 


\subsection{Location and precision pilots}
\label{subsec:pilot-construction}

The preceding theory requires estimates of both the location \(\bmu\)
and the precision matrix \(\bOmega\). We now give a concrete construction
using the spatial median and the spatial-sign version of POET (POET--SS)
proposed by \citet{XuMaWangFeng2025}, and we verify their accuracy under additional
sparsity and growth conditions. For each fold \(I_\ell\), write
\(N_{-\ell}=n-N_\ell\asymp n\).

First, we estimate the location on the training sample by the spatial median
\begin{equation}\label{eq:foldwise-spatial-median}
 \widehat{\bmu}_{-\ell}
 =\arg\min_{\ba\in\R^d}
 \sum_{i\notin I_\ell}\|\bx_i-\ba\|_2,
\end{equation}
with a measurable choice of minimizer if it is not unique. This estimator
does not require a finite covariance matrix. It is used both to center
the training observations when constructing the precision pilot below
and to center the held-out observations in
\eqref{eq:cross-fitted-tyler}.

Next, we define the spatial sign by \(\spSign(\by)=
\by/\|\by\|_2\) for \(\by\ne\boldsymbol0\),
and \(\spSign(\boldsymbol0)=\boldsymbol0\). 
The scaled population spatial-sign covariance matrix is \(\bS=d\E\{\spSign(\bx-\bmu)\spSign(\bx-\bmu)^\top\}\).
Estimate it on the training sample by
\begin{equation}
 \widehat{\bS}_{-\ell}
 =\frac{d}{N_{-\ell}}\sum_{i\notin I_\ell}
 \spSign(\bx_i-\widehat{\bmu}_{-\ell})
 \spSign(\bx_i-\widehat{\bmu}_{-\ell})^\top.
\end{equation}
Let \(\widehat{\bGamma}_{m,-\ell}\in\R^{d\times m}\) contain its leading \(m\) 
orthonormal eigenvectors, and \(\widehat{\bLambda}_{m,-\ell}\) be the diagonal matrix of the
corresponding eigenvalues. Following the POET construction
\citep{fan2013large}, we retain this leading component and threshold the off-diagonal entries of its residual. 
Specifically, we define the residual by
\(\widehat{\bS}_{\varepsilon,-\ell}=\widehat{\bS}_{-\ell}
  -\widehat{\bGamma}_{m,-\ell}
   \widehat{\bLambda}_{m,-\ell}
   (\widehat{\bGamma}_{m,-\ell})^\top\) and write its thresholded version as
\(\widehat{\bH}_{\varepsilon,-\ell}\), with entries
\begin{equation}
 \begin{aligned}
  (\widehat{\bH}_{\varepsilon,-\ell})_{ps}
 =\begin{cases}
   (\widehat{\bS}_{\varepsilon,-\ell})_{pp},&p=s,\\
   \mathcal T_{\tau_{-\ell}}((\widehat{\bS}_{\varepsilon,-\ell})_{ps}),&p\ne s,
  \end{cases}
 \end{aligned}
\end{equation}
where \(\mathcal T_\tau(x)=0\) for \(|x|\le\tau\) and
\(|\mathcal T_\tau(x)-x|\le\tau\) otherwise; hard and soft thresholding are examples.
For the threshold level \(\tau_{-\ell}\), we use \(\tau_{-\ell}=C\delta(N_{-\ell})\), where \(C\) is a sufficiently large constant, and \(\delta(t)= ({\log d}/{t})^{1/2}+({\log t}/{t})^{1/2}+d^{-1/2}\) for a sample size \(t\ge2\).
After that, we combine the leading component
with the thresholded residual to form \(\widehat{\bH}_{-\ell}=\widehat{\bGamma}_{m,-\ell}
  \widehat{\bLambda}_{m,-\ell}
  (\widehat{\bGamma}_{m,-\ell})^\top
  +\widehat{\bH}_{\varepsilon,-\ell}\),
and apply a diagonal adjustment before inversion to ensure positive definiteness of the final precision pilot \(\widehat{\bP}_{-\ell}\):
\begin{equation}\label{eq:foldwise-poet-ss}
 \begin{gathered}
 \varpi_{-\ell}
 =\left\{d^{-1}-\lambda_{\min}
  (\widehat{\bH}_{-\ell})\right\}_+,
 \qquad
 \widehat{\bP}_{-\ell}
 =\left(\widehat{\bH}_{-\ell}
  +\varpi_{-\ell}\bI_d\right)^{-1}.
 \end{gathered}
\end{equation}
The diagonal adjustment guarantees positive
definiteness in finite samples and is inactive with probability tending to one under the conditions of Proposition~\ref{prop:nuisance-construction} below. 

Finally, we verify that the constructed pilots satisfy the required rates in Assumption~\ref{ass:nuisance-estimation}.
For this construction, fix \(v\in[0,1/2)\) and quantify the row sparsity of the idiosyncratic matrix by
\[
 s_d=
 \begin{cases}
  \displaystyle\max_{1\le p\le d}\|(\bSigma_{\varepsilon})_{p\cdot}\|_0,&v=0,\\[3pt]
  \displaystyle\max_{1\le p\le d}\|(\bSigma_{\varepsilon})_{p\cdot}\|_v^v,&0<v<1/2.
 \end{cases}
\]
Thus, \(v=0\) measures exact row sparsity, while \(v>0\) allows many small nonzero entries.
The following proposition connects this construction to the general
nuisance conditions in Section~\ref{subsec:single-sample-asymptotics}.

\begin{proposition}[Accuracy of the nuisance estimators]
\label{prop:nuisance-construction}
Suppose Assumptions~\ref{ass:factor-structure} and
\ref{ass:inverse-radial-moments} hold. If
\(s_d^2\{\log(dn)\}^{1-v}=o(n^{1/2-v})\) and
 \(s_d^2 n^{1/2}=o(d^{1-v})\), then the estimates in
\eqref{eq:foldwise-spatial-median} and \eqref{eq:foldwise-poet-ss}
satisfy Assumption~\ref{ass:nuisance-estimation} with \(a_n=s_d\{\delta(n)\}^{1-v}+d^{-1} = o(n^{-1/4})\) and \(c_d=d\E(Q^{-1})\), where \(Q=\bz^\top\bSigma\bz\) with \(\bz\sim N_d(\boldsymbol0,\bI_d)\).
Moreover, \(\Pp(\varpi_{-\ell}=0\text{ for all }\ell)\to1\).
\end{proposition}

These growth conditions are used to verify the precision-pilot rate; for the spatial median, Assumptions~\ref{ass:factor-structure}--\ref{ass:inverse-radial-moments} imply \(\zeta_1\|\widehat{\bmu}_{-\ell}-\bmu\|_2=O_{\Pp}(n^{-1/2})\), without additional restrictions on the relative growth of \(n\) and \(d\). 
When \(d\asymp n\) and \(s_d=O(1)\), the rate conditions in Proposition~\ref{prop:nuisance-construction} are satisfied for every fixed \(0\le v<1/2\), allowing both exact and approximate sparsity.


\section{Two-Sample Tests}
\label{sec:two-sample-tests}

For the two independent samples in Section~\ref{subsec:testing-problems}, we apply the construction \eqref{eq:pre-poet-orthogonalization} in
Section~\ref{subsec:orth-estimator} separately to each group. Write
\(\widehat\lambda_j^{(r)}\) and
\(\widehat{\bu}_j^{(r)}\) for the resulting estimated
eigenvalues and eigenvectors, where \(r=1,2\). We consider a fixed
\(j\le m=\min(m_1,m_2)\) throughout.

\subsection{Testing principal eigenvalues}
\label{subsec:two-sample-eigenvalues}

To test \(\mathbb{H}_{0,\lambda,j}\) in \eqref{eq:null-eigenvalue}, we compare the two estimated eigenvalues on the logarithmic scale, which preserves the null hypothesis. 
Theorem~\ref{thm:one-sample-eigenvalue}, together with the delta method, gives the asymptotic variance \((\sigma_{\lambda,j}^{(r)})^{2}/n_r\) for \(\log\widehat\lambda_j^{(r)}\), where \((\sigma_{\lambda,j}^{(r)})^{2}\) denotes the population variance in
\eqref{eq:eigenvalue-clt-main} for group \(r\).
Then, the group-specific variance is estimated by
\begin{equation}\label{eq:eigenvalue-variance-estimator}
 (\widehat\sigma_{\lambda,j}^{(r)})^{2}
 =\frac{2(d+2)}{d^3}
  \left\{d^2-2d\widehat\lambda_j^{(r)}
  +\sum_{k=1}^d(\widehat\lambda_k^{(r)})^2\right\},
 \qquad r=1,2.
\end{equation}
Since the two samples are independent, we naturally have the statistic
\begin{equation}
 Z_{\lambda,j}
 =\frac{\log\widehat\lambda_j^{(1)}
       -\log\widehat\lambda_j^{(2)}}
 {\left((\widehat\sigma_{\lambda,j}^{(1)})^{2}/n_1
       +(\widehat\sigma_{\lambda,j}^{(2)})^{2}/n_2\right)^{1/2}}.
\end{equation}
Accordingly, we have the following theorem for the two-sample spiked eigenvalue test.

\begin{theorem}[Two-sample spiked eigenvalue test]
\label{thm:two-sample-eigenvalue-test}
Suppose Assumptions~\ref{ass:factor-structure}--\ref{ass:nuisance-estimation}
hold for each group as
\(n_1,n_2,d\to\infty\). Then
\((\widehat\sigma_{\lambda,j}^{(r)})^{2}/(\sigma_{\lambda,j}^{(r)})^{2}
\xrightarrow{\Pp}1\) for \(r=1,2\). Under \(\mathbb{H}_{0,\lambda,j}\),
\begin{equation}\label{eq:two-sample-eigenvalue-null}
 Z_{\lambda,j}\xrightarrow{\mathcal D}N(0,1).
\end{equation}
\end{theorem}

By Theorem~\ref{thm:two-sample-eigenvalue-test}, we reject \(\mathbb{H}_{0,\lambda,j}\) when \(|Z_{\lambda,j}|>z_{1-\alpha/2}\) for a significance level \(\alpha\in(0,1)\), where \(z_q=\Phi^{-1}(q)\) is the \(q\)th quantile of the standard normal distribution.

\subsection{Testing principal eigendirections}
\label{subsec:two-sample-eigenvectors}

To test \(\mathbb{H}_{0,u,j}\) in \eqref{eq:null-eigenvector}, we compare the two estimated eigendirections through the squared sine of their angle.
Let \(n_{12}=(n_1^{-1}+n_2^{-1})^{-1}\). 
Under the null hypothesis, we align the population eigenvectors so that \(\bu_j^{(1)}=\bu_j^{(2)}=\bu_j\), and write \(\mathbf e_j^{(r)}=\widehat{\bu}_j^{(r)}-\bu_j\), with the estimated eigenvectors aligned accordingly. 
Theorem~\ref{thm:one-sample-eigenvector} and unit normalization give
\[
 n_{12}\left[1-
  \left|\left\langle\widehat{\bu}_j^{(1)},
                         \widehat{\bu}_j^{(2)}
  \right\rangle\right|^2\right]
 =n_{12}\|\mathbf e_j^{(1)}-\mathbf e_j^{(2)}\|_2^2
  +o_{\Pp}(1).
\]
Thus, the two-sample discrepancy is determined by the difference between the estimation errors characterized in Theorem~\ref{thm:one-sample-eigenvector}.
For each group, Theorem~\ref{thm:one-sample-eigenvector} identifies
\(b_{u,j}^{(r)}\) as the nonspiked contribution to the squared estimation error scaled by \(n_r\). Removing these contributions gives the centered statistic
\begin{equation}
 T_{u,j}
 =n_{12}\left[1-
  \left|\left\langle\widehat{\bu}_j^{(1)},
                         \widehat{\bu}_j^{(2)}
  \right\rangle\right|^2\right]
  -\sum_{r=1}^2\frac{n_{12}}{n_r}\widehat b_{u,j}^{(r)},
\end{equation}
where \(\widehat b_{u,j}^{(r)} = {(d+2)}{d}^{-1}\{\widehat\lambda_j^{(r)}\}^{-1}\{d-\sum_{h=1}^{m_r}\widehat\lambda_h^{(r)}\}\) is a consistent estimator of \(b_{u,j}^{(r)}\), as illustrated in Section~\ref{subsec:single-sample-asymptotics}.

Nevertheless, the asymptotic null distribution of \(T_{u,j}\) still depends on the unknown spiked eigenvalues and the relative orientations of the other spiked eigenvectors across the two populations. We therefore construct a plug-in reference distribution to obtain critical values for the test.
Specifically, let
\(\Kcal_j^{(r)}=\{1,\ldots,m_r\}\setminus\{j\}\), and write
\(m_u=m_1+m_2-2\) for the combined number of spiked directions after
excluding the target direction in each group. We estimate the leading-direction variances in
\eqref{eq:direction-population-weight} by
\begin{equation}
 (\widehat\sigma_{u,jk}^{(r)})^{2}
 =\frac{d+2}{d}
  \frac{\widehat\lambda_j^{(r)}
        \widehat\lambda_k^{(r)}}
  {(\widehat\lambda_j^{(r)}
   -\widehat\lambda_k^{(r)})^2},
 \qquad r=1,2,\quad k\in\Kcal_j^{(r)}.
\end{equation}
Let
\(\widehat{\bGamma}_{-j}^{(r)}\) collect
\(\widehat{\bu}_k^{(r)}\) for \(k\in\Kcal_j^{(r)}\) in increasing order, and set
\(\widehat{\boldsymbol\Theta}_{j,d}
 =\bigl(\langle\widehat{\bu}_k^{(1)},\widehat{\bu}_l^{(2)}\rangle\bigr)_{
 k\in\Kcal_j^{(1)},\ l\in\Kcal_j^{(2)}}\).
We define
\begin{equation}\label{eq:plugin-direction-matrices}
 \begin{aligned}
 \widehat{\bA}_j
 &=\begin{pmatrix}
   \bI_{m_1-1}&-\widehat{\boldsymbol\Theta}_{j,d}\\
   -\widehat{\boldsymbol\Theta}_{j,d}^\top&\bI_{m_2-1}
  \end{pmatrix},\qquad \widehat{\bC}_j =\Diag\left(
   \left(\frac{n_{12}}{n_1}(\widehat\sigma_{u,jk}^{(1)})^{2}
   \right)_{k\in\Kcal_j^{(1)}},
   \left(\frac{n_{12}}{n_2}(\widehat\sigma_{u,jk}^{(2)})^{2}
   \right)_{k\in\Kcal_j^{(2)}}\right).
 \end{aligned}
\end{equation}
Note that the matrix \(\widehat{\bA}_j\) is the Gram matrix of
\([\widehat{\bGamma}_{-j}^{(1)},
-\widehat{\bGamma}_{-j}^{(2)}]\), and hence is positive
semidefinite. 
As in Theorem~\ref{thm:one-sample-eigenvector}, we use a Gaussian
quadratic-form reference, whose weights are now determined by both
\(\widehat{\bA}_j\) and \(\widehat{\bC}_j\). Conditionally on the data,
\begin{equation}\label{eq:plugin-direction-reference}
 \widehat Q_j=\bz^\top\widehat{\bK}_{u,j}\bz
 \stackrel{\mathcal D}{=}
 \sum_{\ell=1}^{m_u}
  \widehat\gamma_{j\ell}Z_\ell^2,
\end{equation}
where \(\bz=(Z_1,\ldots,Z_{m_u})^\top\sim N_{m_u}(\boldsymbol0,\bI_{m_u})\)
is independent of the data, and \(\widehat\gamma_{j1},\ldots,
\widehat\gamma_{j m_u}\) are the nonnegative eigenvalues of
\(\widehat{\bK}_{u,j}=\widehat{\bC}_j^{1/2}
\widehat{\bA}_j\widehat{\bC}_j^{1/2}\).

The following theorem establishes the validity of this plug-in reference distribution for the two-sample spiked eigendirection test.

\begin{theorem}[Two-sample spiked eigendirection test]
\label{thm:two-sample-eigendirection-test}
Let \(\widehat F_{u,j}(x)=\Pp(\widehat Q_j\le x\mid\text{data})\). Suppose Assumptions~\ref{ass:factor-structure}--\ref{ass:nuisance-estimation}
hold for each group as
\(n_1,n_2,d\to\infty\), and
\(n_{12}/n_r\to\pi^{(r)}\in(0,1)\), \(r=1,2\).
If \(m_u\ge1\), then under \(\mathbb{H}_{0,u,j}\),
\begin{equation}\label{eq:two-sample-eigendirection-calibration}
 \sup_{x\in\R}
 \left|\Pp(T_{u,j}\le x)-\widehat F_{u,j}(x)\right|
 \xrightarrow{\Pp}0.
\end{equation}
\end{theorem}

From Theorem~\ref{thm:two-sample-eigendirection-test}, we reject
\(\mathbb{H}_{0,u,j}\) when \(T_{u,j}>\widehat q_{u,j}(1-\alpha)\), where \(\widehat q_{u,j}(1-\alpha)\) is the \((1-\alpha)\)th quantile of \(\hat{F}_{u,j}\). 
The reference quantile and tail probability can be evaluated numerically from the fitted weights in \eqref{eq:plugin-direction-reference}.
Theorem~\ref{thm:two-sample-eigendirection-test} implies that the test has asymptotic level \(\alpha\).
It is also worth noting that both the statistic \(T_{u, j}\) and its reference distribution \(\hat{F}_{u, j}\) are invariant to eigenvector sign choices. 
When \(m_1=m_2=1\), the leading random component is absent and \(T_{u,1}\xrightarrow{\Pp}0\) under the null; this degenerate case
is excluded from the calibration above.

\begin{remark}[Estimated factor numbers]\label{rem:estimated-factor-numbers}
Under the respective assumptions of
Theorems~\ref{thm:two-sample-eigenvalue-test} and
\ref{thm:two-sample-eigendirection-test}, both tests retain asymptotic
level \(\alpha\) after replacing \(m_r\) by
\(\widehat m_r\) throughout nuisance estimation and test calibration,
provided that
\[
 \Pp\{(\widehat m_1,\widehat m_2)=(m_1,m_2)\}\longrightarrow1.
\]
By convention, we do not reject if \(j>\min(\widehat m_1,\widehat m_2)\),
or, for the eigendirection test, if \(\widehat m_1+\widehat m_2-2=0\).
A short justification is given in
Section~\ref{supp:subsec:section-4-2-proofs}.
\end{remark}

\section{Numerical Studies}
\label{sec:simulation}


\subsection{Simulation design}
\label{subsec:simulation-design}

In the main experiments, we generate two independent samples from the
elliptical model in \eqref{eq:model}, with zero locations and factor ranks
\(m_1=m_2=3\):
\[
 \bx_i^{(r)}=\xi_i^{(r)}(\bSigma^{(r)})^{1/2}\bs_i^{(r)},
 \qquad \bs_i^{(r)}\sim\operatorname{Unif}(\mathbb S^{d-1}),
 \qquad r=1,2.
\]
The radial variable \(\xi_i^{(r)}\) is independent of \(\bs_i^{(r)}\),
and observations are independent within each group. We specify the
shape matrices and radial distributions below.
The shape matrix is specified as \(\bSigma^{(r)} =a^{(r)}\{\bL^{(r)}(\bL^{(r)})^\top+\bR\}\), where \(\bL^{(r)}\in\R^{d\times 3}\) is the group-specific loading matrix specified later, and \(\bR\) has entries \(R_{pq}=0.5^{|p-q|}\). The normalization constant \(a^{(r)}=d \tr^{-1}\{\bL^{(r)}(\bL^{(r)})^\top+\bR\}\) ensures that \(\tr(\bSigma^{(r)})=d\).
We consider the Gaussian, elliptical \(t_8\), and elliptical \(t_3\) distributions, with the corresponding radial variables specified by
\[
 \xi_i^{(r)}=
 \begin{cases}
  \sqrt{W_i^{(r)}},&\text{Gaussian},\\[3pt]
  \sqrt{(\nu-2)W_i^{(r)}/V_i^{(r)}},
   &\text{elliptical }t_\nu,\quad\nu\in\{8,3\},
 \end{cases}
\]
where \(W_i^{(r)}\sim\chi_d^2\) and \(V_i^{(r)}\sim\chi_\nu^2\) are independent of each other and of \(\bs_i^{(r)}\).
The elliptical \(t_8\) distribution has finite fourth moments, whereas the
elliptical \(t_3\) distribution has finite second but infinite fourth moments.
Both groups use the same radial family in each experiment, and these choices satisfy \(\E\{(\xi_i^{(r)})^2\}=d\), so that \(\Cov(\bx_i^{(r)})=\bSigma^{(r)}\).

For the null hypothesis, we set \(\bL^{(1)}=\bL^{(2)}=\bL\), so that \(\bSigma^{(1)}=\bSigma^{(2)}\), where \(\bL\in\R^{d\times3}\) has independent entries \(L_{pk}\sim\operatorname{Unif}(-\sqrt3c_k,\sqrt3c_k)\) with \((c_1,c_2,c_3)=(1,0.75,0.50)\).
For the alternatives, we keep \(\bL^{(1)}=\bL\) and write the compact singular-value decomposition
\[
 \bL=\bQ\bD^{1/2}\mathbf V^\top,
 \qquad \bD=\Diag(\beta_1,\beta_2,\beta_3),
 \qquad \beta_1\ge\beta_2\ge\beta_3>0,
\]
where \(\bQ\in\R^{d\times3}\) and \(\mathbf V\in\R^{3\times3}\) have orthonormal columns, and the \(\beta_k\) are the nonzero eigenvalues of \(\bL\bL^\top\). 
We then set
\[
 \bL^{(2)}=
 \begin{cases}
  \bQ\Diag(\sqrt{h\beta_1},\sqrt{\beta_2},\sqrt{\beta_3})\mathbf V^\top,
   &\text{eigenvalue alternative},\quad h>1,\\[4pt]
  \bQ\mathbf O_\theta\bD^{1/2}\mathbf V^\top,
   &\text{eigendirection alternative},
 \end{cases}
\]
where
\[
 \mathbf O_\theta=
 \begin{pmatrix}
  \cos\theta&-\sin\theta&0\\
  \sin\theta&\cos\theta&0\\
  0&0&1
 \end{pmatrix}.
\]
The first construction for testing eigenvalues replaces \(\beta_1\) by \(h\beta_1\) while
preserving the factor eigendirections; the second for testing eigendirections rotates the first two
factor eigendirections while preserving \(\beta_1,\beta_2,\beta_3\).
We measure the resulting discrepancies using the full normalized shapes:
\begin{equation}\label{eq:simulation-effects}
 \Delta_\lambda
 =|\log\lambda_1^{(1)}-\log\lambda_1^{(2)}|,
 \qquad
 \Delta_u
 =\frac{180}{\pi}
   \arccos|\langle\bu_1^{(1)},\bu_1^{(2)}\rangle|.
\end{equation}

For our testing procedure, we use the spatial-median location pilot and the POET--SS precision pilot of \citet{XuMaWangFeng2025} for each group, as introduced in Section~\ref{subsec:pilot-construction}. We use \(K=3\) balanced cross-fitting folds and hard thresholding in POET--SS with \(C=0.60\).
The benchmarks are the sample-covariance eigenvalue-proportion and
eigendirection tests of \citet{FanLiXiaZheng2026}. 
Both methods set the number of factors to three in each group unless otherwise stated. 
The nominal level is \(0.05\) for all tests. 
The main null experiments use 5,000 replications, and the power experiments use 2,000 replications per point.

\subsection{Empirical size and power}

The results of the null experiments with \(d=n_1=n_2=500\) are displayed in Figure~\ref{fig:qq-t3} and Table~\ref{tab:size-balanced}.
In particular, Figure~\ref{fig:qq-t3} shows the quantile--quantile (Q--Q) plots and histograms for the eigenvalue and eigendirection statistics under the elliptical \(t_3\) distribution for the first principal component. 
We can observe that the proposed statistics closely follow their reference laws (Figure~\ref{fig:qq-t3}(a)--(b) and (e)--(f)), but the benchmark eigenvalue statistic is too concentrated (Figure~\ref{fig:qq-t3}(c)--(d)) and its direction statistic is strongly inflated (Figure~\ref{fig:qq-t3}(g)--(h)).
On the other hand, Table~\ref{tab:size-balanced} reports the empirical rejection probabilities for all three principal components under the Gaussian, elliptical \(t_8\), and elliptical \(t_3\) distributions. 
The proposed tests maintain the nominal size across all three distributions, while the benchmark tests are severely distorted under heavy-tailed distributions, which can be expected since they do not account for the heavy tails in the data in their models. 

\begin{figure}[tb]
\centering
\includegraphics[width=0.85\textwidth]{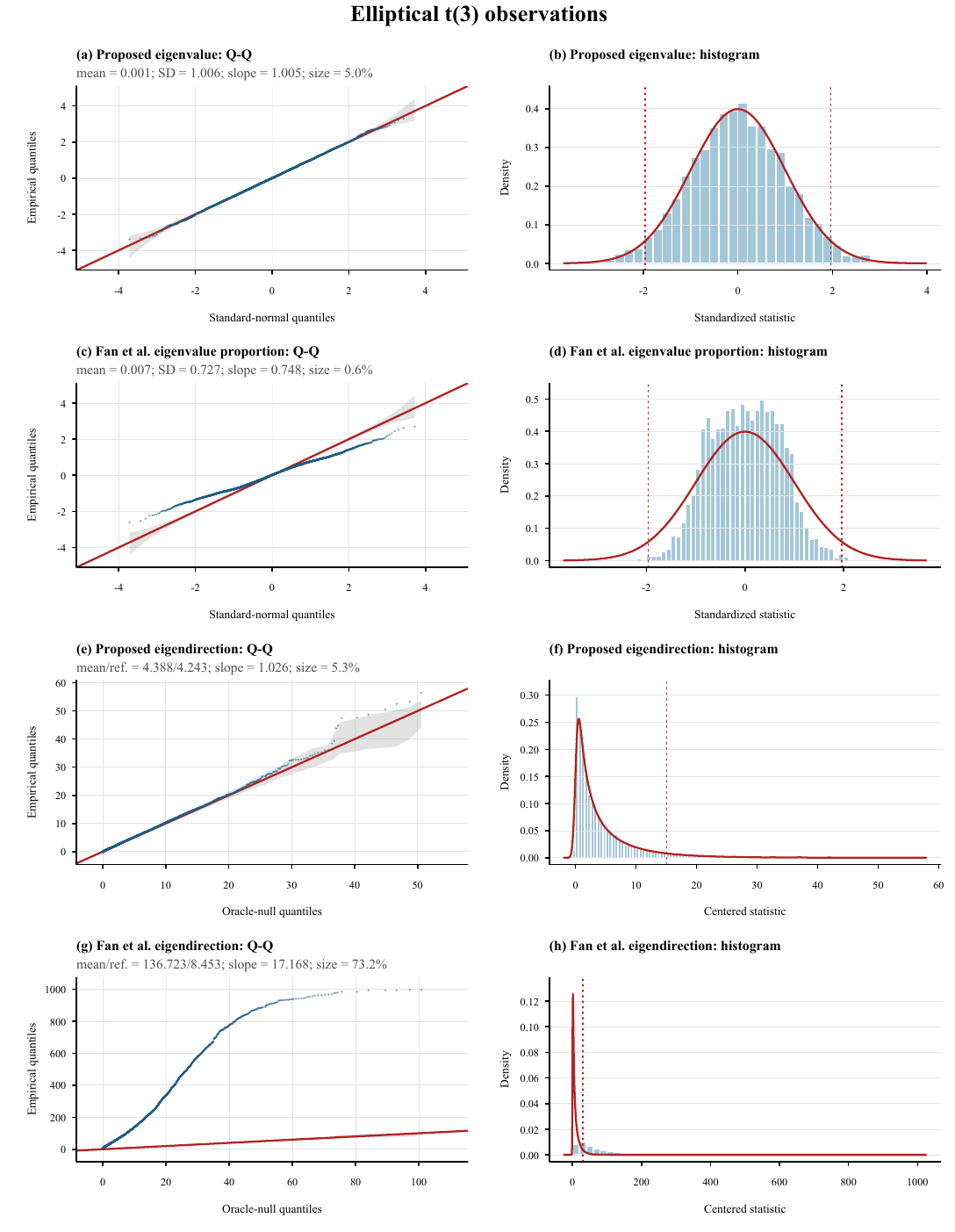}
\caption{Null diagnostics under the elliptical \(t_3\) distribution for the first principal component, with \(d=n_1=n_2=500\), \(m_1=m_2=3\), and 5,000 replications. 
In the Q--Q plots, red lines show \(y=x\), and gray bands are pointwise \(95\%\) reference envelopes for the empirical quantiles.
Each reported Q--Q slope is computed from quantile pairs with probability levels between \(0.05\) and \(0.95\).
In the histograms, red curves show the reference densities, and dashed vertical lines indicate the critical values at significance level \(0.05\).
}
\label{fig:qq-t3}
\end{figure}

\begin{table}[tb]
\centering
\small
\setlength{\tabcolsep}{6.5pt}
\renewcommand{\arraystretch}{1.02}
\begin{threeparttable}
\caption{Empirical rejection probabilities under the null at nominal level
\(0.05\), with \(d=n_1=n_2=500\) and \(m_1=m_2=3\).}
\label{tab:size-balanced}
\begin{tabular}{lcrrrr}
\toprule
& & \multicolumn{2}{c}{Eigenvalue}
& \multicolumn{2}{c}{Eigendirection} \\
\cmidrule(lr){3-4}\cmidrule(lr){5-6}
Distribution & \(j\)
& \multicolumn{1}{c}{Proposed} & \multicolumn{1}{c}{Fan et al.}
& \multicolumn{1}{c}{Proposed} & \multicolumn{1}{c}{Fan et al.} \\
\midrule
\multirow{3}{*}{Gaussian} & 1 & 0.051 & 0.051 & 0.052 & 0.050 \\
 & 2 & 0.051 & 0.052 & 0.052 & 0.047 \\
 & 3 & 0.047 & 0.044 & 0.058 & 0.051 \\
\midrule
\multirow{3}{*}{\(t_8\)} & 1 & 0.049 & 0.032 & 0.051 & 0.117 \\
 & 2 & 0.051 & 0.034 & 0.049 & 0.126 \\
 & 3 & 0.051 & 0.030 & 0.058 & 0.264 \\
\midrule
\multirow{3}{*}{\(t_3\)} & 1 & 0.050 & 0.004 & 0.054 & 0.698 \\
 & 2 & 0.051 & 0.024 & 0.050 & 0.734 \\
 & 3 & 0.049 & 0.026 & 0.061 & 0.995 \\
\bottomrule
\end{tabular}
\end{threeparttable}
\end{table}

For the power experiments, we adopt \(d=n_1=n_2=250\) under the Gaussian and elliptical \(t_3\) distributions and consider the tests for the first principal component. Since the benchmark tests exhibit size distortions under heavy-tailed distributions, we report size-adjusted power for both methods.
Specifically, we obtain method-specific critical values from independent null simulations to target an empirical size of \(0.05\) and then report the
proportion of alternative replications rejected using these critical values.
Figure~\ref{fig:power-adjusted} presents the resulting size-adjusted power comparison. 
Under the Gaussian distribution, the curves of the two methods are nearly indistinguishable, but under the elliptical \(t_3\) distribution, the proposed method clearly has higher adjusted power over the simulated range.
Both the size and power experiments show that the proposed tests are robust to heavy tails and maintain nominal size, which demonstrates their effectiveness as reliable statistical tools.


\begin{figure}[tb]
\centering
\includegraphics[width=0.9\textwidth]{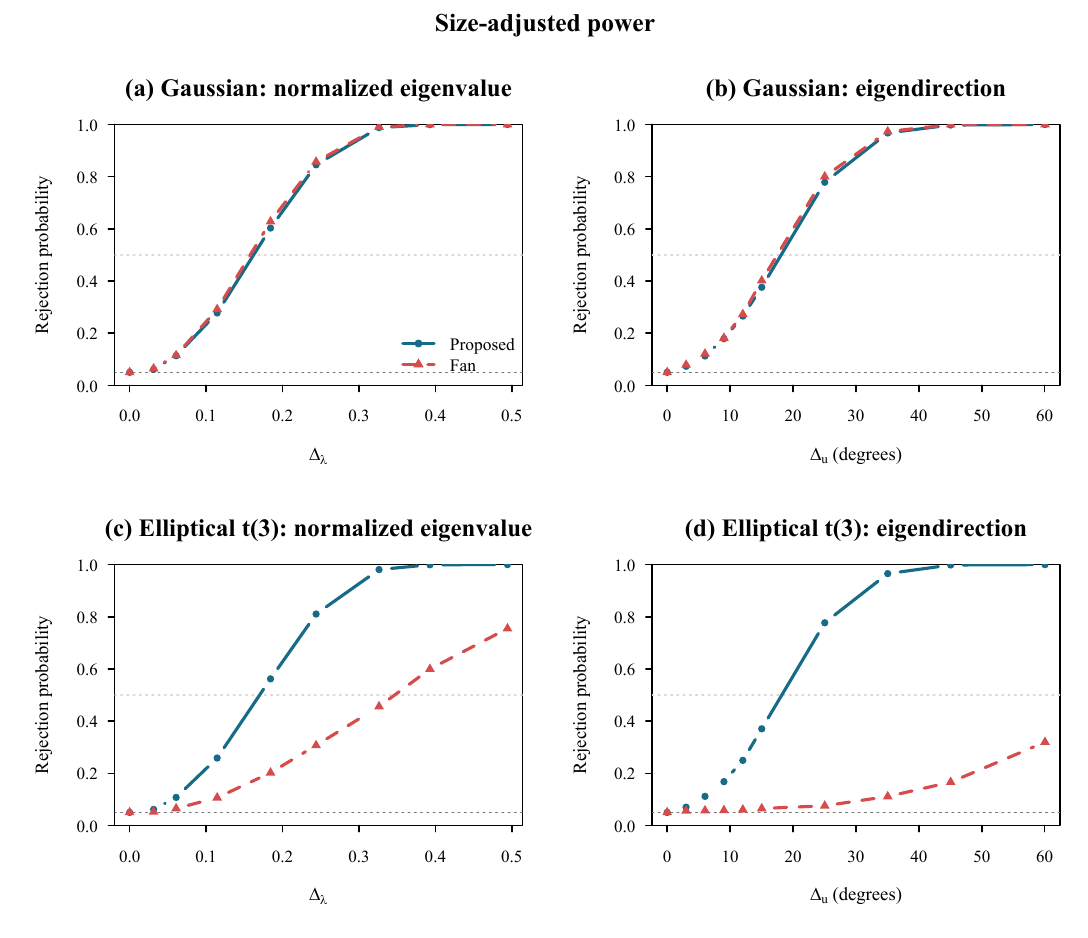}
\caption{Size-adjusted power for the first principal component under the Gaussian and elliptical
\(t_3\) distributions, with \(d=n_1=n_2=250\), \(m_1=m_2=3\), and 2,000
replications per point. 
Blue circles and red triangles denote the proposed method and the benchmark, respectively. The dark dotted line marks the level \(0.05\).}
\label{fig:power-adjusted}
\end{figure}

Additional experiments in Sections~\ref{supp:subsec:factor-sensitivity} and
\ref{supp:subsec:serial-stress} of the Supplementary Material examine the sensitivity to the number of factors and serial departures from the i.i.d. elliptical model. 
\subsection{Application to S\&P 500 stock returns}
\label{sec:application}


We use daily returns from the Center for Research in Security Prices (CRSP) from January 2007 through December 2024.
Following the sampling design of \citet{FanLiXiaZheng2026}, we identify 240 securities whose recorded S\&P 500 membership intervals cover all 4,530 trading days and select a fixed panel of \(d=200\). 
For stock \(p\) on day \(t\), let \(\mathrm{RET}_{pt}\) denote the daily holding-period return reported by CRSP. We use the corresponding log return \(\log(1+\mathrm{RET}_{pt})\) as the observation.
After removing one incomplete date, we get 4,529 daily return vectors. 
Across the 3,600 stock--year return series, the median sample excess kurtosis
is \(2.5\), which suggests
departures from Gaussian tail behavior and motivates inference that
accommodates heavy-tailed returns.
Additional tail diagnostics appear in
Section~\ref{supp:subsec:empirical-tails}.

We consider the break dates reported by \citet{FanLiXiaZheng2026}, which are treated as fixed inputs to this diagnostic analysis.
For each of the seven candidate dates, we form two groups of daily returns using the preceding and following observations, with each adjacent window capped at one year.
The sample sizes \(n_1\) and \(n_2\) range from 194 to 253 trading days.
For each date, we select the factor number by minimizing the information criterion \(IC_{p2}\) of \citet{BaiNg2002} over \(0,\ldots,8\) in each group separately, and denote the selected factor numbers by \(\widehat m_1\) and \(\widehat m_2\).
Then, we apply the proposed two-sample tests to the first three principal components (PC1--PC3) and report the corresponding p-values.


Table~\ref{tab:sp500-break-tests} reports the p-values for the seven candidate dates.
The eigenvalue and eigendirection test p-values for the first three principal components are denoted by \(p_{\lambda,j}\) and \(p_{u,j}\), respectively, for \(j=1,2,3\).
The eigenvalue tests concern relative component strength under the trace constraint, whereas the eigendirection tests concern the stock weights defining each component.
Both tests reject for PC1 at all seven dates at the \(0.05\) level; Figure~\ref{fig:sp500-break-events} places these comparisons along the S\&P 500 index series.
The results for PC2 and PC3 are more varied.
For example, in September 2008, eigendirection equality for PC2 is rejected
(\(p_{u,2}<0.001\)), whereas eigenvalue equality is not
(\(p_{\lambda,2}=0.185\)). The November 2016 comparison gives the opposite pattern:
\(p_{\lambda,2}=0.017\) and \(p_{u,2}=0.292\).
These comparisons illustrate the complementary information provided by
the two tests.

For February 2020, direction equality is rejected for both PC2 and PC3
(\(p_{u,2}<0.001\) and \(p_{u,3}=0.012\)), while neither eigenvalue test
rejects (\(p_{\lambda,2}=0.228\) and \(p_{\lambda,3}=0.583\)).
This comparison spans the onset of the COVID-19 market
disruption and the subsequent recovery, during which technology and
health care stocks outperformed energy and financial stocks
\citep{BIS2020}.
The pandemic's uneven effects across firms may have altered their
comovement patterns and the stock weights defining secondary components.
This offers a possible interpretation of the directional changes detected
here, which the eigenvalue tests alone would not reveal.

\begin{table}[tb]
\centering
\small
\setlength{\tabcolsep}{5pt}
\begin{threeparttable}
\caption{The p-values from the proposed method for testing the eigenvalues and eigendirections of the fixed 200-stock S\&P 500 panel at the seven candidate dates.}
\label{tab:sp500-break-tests}
\begin{tabular}{lcc*{6}{c}}
\toprule
& $\widehat m_1$ & $\widehat m_2$ & \multicolumn{2}{c}{PC1} & \multicolumn{2}{c}{PC2} & \multicolumn{2}{c}{PC3} \\
\cmidrule(lr){4-5}\cmidrule(lr){6-7}\cmidrule(lr){8-9}
Date & & & \(p_{\lambda,1}\) & \(p_{u,1}\) & \(p_{\lambda,2}\) & \(p_{u,2}\) & \(p_{\lambda,3}\) & \(p_{u,3}\) \\
\midrule
2008-09-12 & 4 & 4 & \(\boldsymbol{<0.001}\) & \(\boldsymbol{<0.001}\) & \(0.185\) & \(\boldsymbol{<0.001}\) & \(0.639\) & \(0.430\) \\
2009-06-22 & 4 & 2 & \(\mathbf{0.026}\) & \(\boldsymbol{<0.001}\) & \(\mathbf{0.012}\) & \(\boldsymbol{<0.001}\) & NA & NA \\
2011-12-21 & 3 & 2 & \(\boldsymbol{<0.001}\) & \(\boldsymbol{<0.001}\) & \(\mathbf{0.010}\) & \(\boldsymbol{<0.001}\) & NA & NA \\
2014-10-07 & 2 & 4 & \(\mathbf{0.001}\) & \(\boldsymbol{<0.001}\) & \(\boldsymbol{<0.001}\) & \(\boldsymbol{<0.001}\) & NA & NA \\
2016-11-07 & 4 & 4 & \(\boldsymbol{<0.001}\) & \(\boldsymbol{<0.001}\) & \(\mathbf{0.017}\) & \(0.292\) & \(0.364\) & \(0.350\) \\
2020-02-21 & 4 & 7 & \(\boldsymbol{<0.001}\) & \(\boldsymbol{<0.001}\) & \(0.228\) & \(\boldsymbol{<0.001}\) & \(0.583\) & \(\mathbf{0.012}\) \\
2022-10-17 & 6 & 5 & \(\mathbf{0.035}\) & \(\boldsymbol{<0.001}\) & \(\mathbf{0.005}\) & \(0.152\) & \(0.961\) & \(0.482\) \\
\bottomrule
\end{tabular}
\begin{tablenotes}[flushleft]\footnotesize
\item 
Bold entries indicate rejection at \(0.05\);
NA denotes \(j>\min(\widehat m_1,\widehat m_2)\). 
\end{tablenotes}
\end{threeparttable}
\end{table}

\begin{figure}[tb]
\centering
\includegraphics[width=0.9\textwidth]{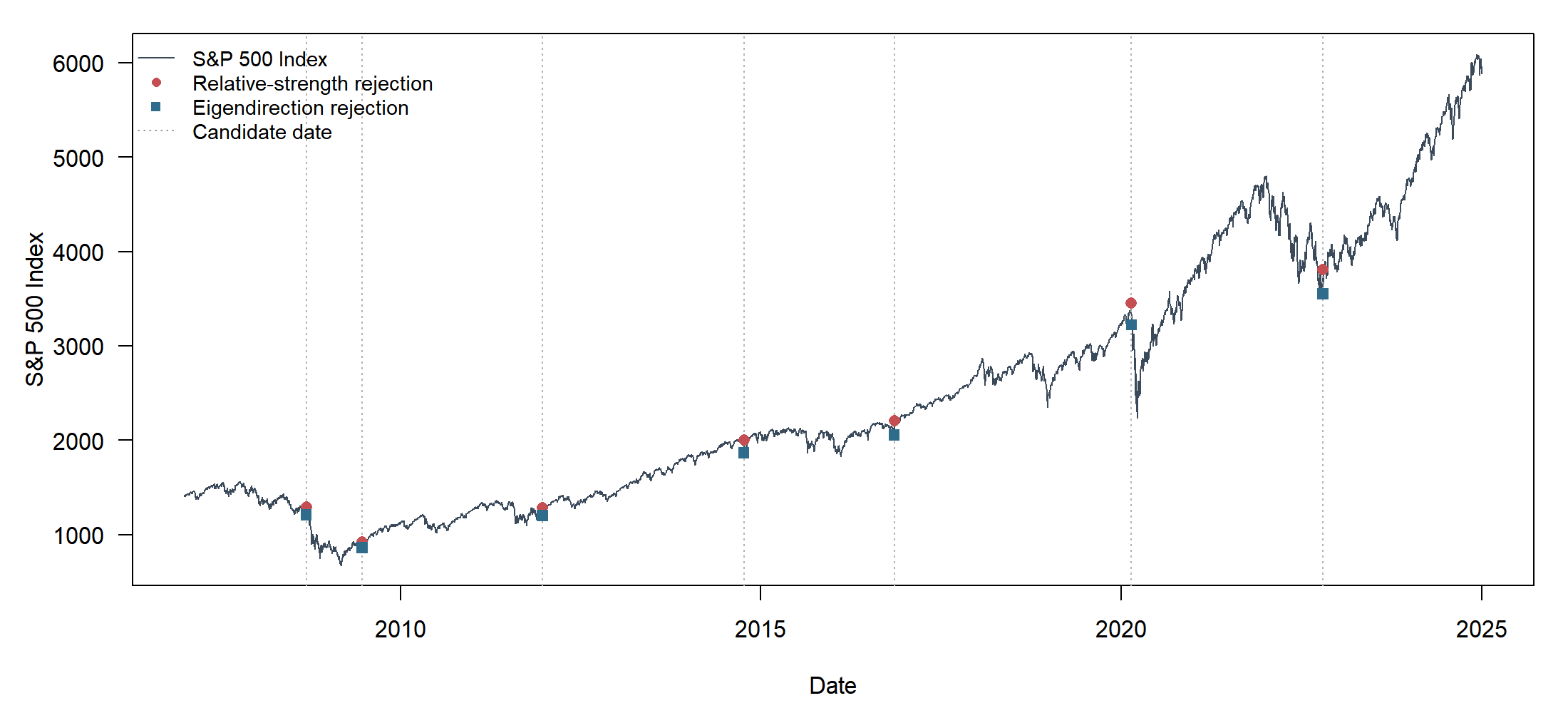}
\caption{S\&P 500 index and PC1 rejections at the seven candidate dates.
Dotted vertical lines mark the dates. Red circles and blue squares
indicate rejection of eigenvalue and eigendirection equality,
respectively, at the level \(0.05\).}
\label{fig:sp500-break-events}
\end{figure}
\section{Concluding Remarks}
\label{sec:conclusion}

We establish asymptotically valid two-sample inference for principal eigenvalues and eigendirections under high-dimensional elliptical factor models. The proposed framework accommodates different radial laws across samples and general
pilots satisfying explicit accuracy conditions, enabling separate
diagnostics for changes in relative component strength and composition.

Several extensions merit further study. First, allowing temporal dependence would require controlling dependence between nuisance fitting and score evaluation, estimating the long-run covariance of the angular scores, and revisiting the accumulated nonspiked contribution. 
Second, extending the pointwise tests to simultaneous inference over components and time windows could support change-point analysis. 
This would require joint approximations and calibration that account for searching over candidate dates.
Finally, inference for principal eigenspaces could accommodate repeated or closely spaced spiked eigenvalues, for which individual eigendirections may be unidentifiable or difficult to estimate reliably.


\clearpage
\pdfbookmark[0]{Supplementary Material}{supplement-title}
\begin{center}
{\Large\bfseries Supplementary Material for
{\Large ``\papertitle''}\par}
\end{center}
\vspace{1em}

\setcounter{section}{0}
\setcounter{subsection}{0}
\setcounter{subsubsection}{0}
\setcounter{equation}{0}
\setcounter{theorem}{0}
\setcounter{assumption}{0}
\setcounter{lemma}{0}
\setcounter{proposition}{0}
\setcounter{remark}{0}
\setcounter{figure}{0}
\setcounter{table}{0}
\renewcommand{\thesection}{\Alph{section}}
\renewcommand{\thesubsection}{\thesection.\arabic{subsection}}
\counterwithin{equation}{section}
\counterwithin{theorem}{section}
\counterwithin{assumption}{section}
\counterwithin{lemma}{section}
\counterwithin{proposition}{section}
\counterwithin{remark}{section}
\counterwithin{figure}{section}
\counterwithin{table}{section}

\providecommand*{\theHsection}{}
\providecommand*{\theHsubsection}{}
\providecommand*{\theHsubsubsection}{}
\providecommand*{\theHequation}{}
\providecommand*{\theHtheorem}{}
\providecommand*{\theHassumption}{}
\providecommand*{\theHlemma}{}
\providecommand*{\theHproposition}{}
\providecommand*{\theHremark}{}
\providecommand*{\theHfigure}{}
\providecommand*{\theHtable}{}
\renewcommand*{\theHsection}{supp.\arabic{section}}
\renewcommand*{\theHsubsection}{supp.\arabic{section}.\arabic{subsection}}
\renewcommand*{\theHsubsubsection}{supp.\arabic{section}.\arabic{subsection}.\arabic{subsubsection}}
\renewcommand*{\theHequation}{supp.\arabic{section}.\arabic{equation}}
\renewcommand*{\theHtheorem}{supp.\arabic{section}.\arabic{theorem}}
\renewcommand*{\theHassumption}{supp.\arabic{section}.\arabic{assumption}}
\renewcommand*{\theHlemma}{supp.\arabic{section}.\arabic{lemma}}
\renewcommand*{\theHproposition}{supp.\arabic{section}.\arabic{proposition}}
\renewcommand*{\theHremark}{supp.\arabic{section}.\arabic{remark}}
\renewcommand*{\theHfigure}{supp.\arabic{section}.\arabic{figure}}
\renewcommand*{\theHtable}{supp.\arabic{section}.\arabic{table}}

\section{Proofs}
\label{supp:sec:proofs}

\subsection{Proofs for section~\ref*{sec:single-sample}}
\label{supp:subsec:section-3-proofs}
We use the notation of Section~\ref{sec:single-sample}. Unless stated
otherwise, Assumptions~\ref{ass:factor-structure}--\ref{ass:nuisance-estimation}
hold, and \(C\) denotes a finite constant
independent of \(n,d\). In particular,
\[
 \|\bOmega\|_{\op}\le c_{\varepsilon}^{-1},\qquad
 \|\bSigma\|_{\mathrm F}\le d,\qquad
 \lambda_j\asymp g_j\asymp d\quad(j\le m),
 \qquad g_j=\min_{k\ne j}|\lambda_j-\lambda_k|.
\]
Throughout this subsection, \(a_n\) denotes the common bound on the
precision and standardized location errors in
Assumption~\ref{ass:nuisance-estimation}. Recall
\(\epsilon_n=a_n+n^{-1/2}\), \(\rho_n=\epsilon_n^2\), and
\(\sqrt n\,\rho_n\to0\).

Write \(\bx_i=\bmu+\xi_i\bSigma^{1/2}\bs_i\) for the sample
representation. For \(i\in I_\ell\), set
\[
 \by_i=\bx_i-\bmu,\qquad
 \bdelta_{\mu,-\ell}=\widehat{\bmu}_{-\ell}-\bmu,\qquad
 \bDelta_{P,-\ell}=c_d\widehat{\bP}_{-\ell}-\bOmega,
 \qquad
 \bar{\bDelta}_P=\sum_{\ell=1}^K\frac{N_\ell}{n}\bDelta_{P,-\ell}
 =c_d\bar{\bP}-\bOmega.
\]
Let
\(\mathcal F_{-\ell}=\sigma(I_1,\ldots,I_K)\vee
\sigma(\bx_i:i\notin I_\ell)\).
Conditionally on \(\mathcal F_{-\ell}\), the nuisance errors are fixed and
the observations in \(I_\ell\) are independent and retain their population
distribution. All conditional arguments are applied one fold at a time;
the foldwise results are then combined over the fixed number \(K\) of folds.

It is convenient to control the Frobenius norm and the trace together.
For a matrix \(\bH\), define the Hilbert norm
\begin{equation}\label{supp:eq:joint-matrix-norm}
 \|\bH\|_\diamond
 =\frac1d\{\|\bH\|_{\mathrm F}^2+|\tr \bH|^2\}^{1/2}.
\end{equation}
For vectors \(\bv,\bw\), \(\|d\bv\bw^\top\|_\diamond\le\sqrt2\|\bv\|_2\|\bw\|_2\).
This auxiliary norm makes the trace remainder explicit before
normalization. We first state the supporting lemmas; their proofs are
collected in Section~\ref{supp:subsec:auxiliary-lemma-proofs}.

\begin{lemma}[Angular moments and denominator control]
\label{supp:lem:angular-denominators}
Put \(\by=\xi\bSigma^{1/2}\bs\),
\(\mathcal A=\bs^\top\bSigma\bs\), and
\(\omega=(\zeta_1\xi)^{-1}\), with sample counterparts
\(\mathcal A_i=\bs_i^\top\bSigma\bs_i\) and
\(\omega_i=(\zeta_1\xi_i)^{-1}\). For every fixed positive integer \(q\),
\[
 \E\mathcal A^q\le(2q-1)!!,\qquad
 \E\omega^q\le C_q.
\]
Moreover,
\[
 \max_{i\le n}\mathcal A_i=O_{\Pp}(\log n),\qquad
 \max_{i\le n}(1+\mathcal A_i^{1/2})\omega_i=O_{\Pp}(n^{1/4}).
\]
Define
\[
 e_{\ell i}
 =\frac{(\by_i-\bdelta_{\mu,-\ell})^\top
       (\bOmega+\bDelta_{P,-\ell})(\by_i-\bdelta_{\mu,-\ell})}
       {\xi_i^2}-1.
\]
Then
\begin{equation}\label{supp:eq:uniform-denominator-control}
 \max_{\ell\le K}\max_{i\in I_\ell}|e_{\ell i}|
 =O_{\Pp}\{a_n\log n+a_n n^{1/4}+a_n^2 n^{1/2}\}
 =o_{\Pp}(1).
\end{equation}
In particular, all these relative denominators lie in \([1/2,3/2]\)
with probability tending to one.
\end{lemma}

Define the rescaled Tyler update and the centered oracle score by
\[
 \bM_n
 =c_d^{-1}\widehat{\bSigma}_{\mathrm{cf}},\qquad
 \bPhi(\bx)=d\bSigma^{1/2}\bs\bs^\top\bSigma^{1/2}-\bSigma.
\]
Here \(\bM_n\) is the Tyler update obtained after multiplying
every precision pilot by \(c_d\); it has not been trace-normalized.
This rescaling is used only in the proof. By
Remark~\ref{rem:pilot-scale}, it leaves \(\widetilde{\bSigma}\),
\(\widetilde{\bOmega}\), and the final estimator unchanged.

\begin{lemma}[Joint expansion of the rescaled Tyler update]
\label{supp:lem:raw-linear}
There exists a remainder \(\bR_{\mathrm{cf},n}\) such that
\begin{equation}\label{supp:eq:raw-matrix-expansion}
 \begin{aligned}
 \bM_n-\bSigma
 ={}&\frac1n\sum_{i=1}^n\bPhi(\bx_i)\\
 &-\frac{2\bSigma\bar{\bDelta}_P\bSigma
       +\tr(\bSigma\bar{\bDelta}_P)\bSigma}{d+2}
   +\bR_{\mathrm{cf},n},
 \qquad \|\bR_{\mathrm{cf},n}\|_\diamond=O_{\Pp}(\rho_n).
 \end{aligned}
\end{equation}
Moreover,
\(\|\bM_n-\bSigma\|_\diamond
=O_{\Pp}(\epsilon_n)\).
The population location derivative is zero, while the population
precision derivative in the direction \(\bH=\bH^\top\) is
\begin{equation}\label{supp:eq:raw-tyler-derivative}
 -\frac{2\bSigma \bH\bSigma+\tr(\bSigma \bH)\bSigma}{d+2}.
\end{equation}
\end{lemma}

\begin{lemma}[Simple-eigenpair perturbation]
\label{supp:lem:spectral-perturbation}
Let \(\bE=\bE^\top\) and \(\|\bE\|_{\op}\le g_j/4\). Denote by
\((\widetilde\lambda_j,\widetilde{\bu}_j)\) the corresponding eigenpair
of \(\bSigma+\bE\), with
\(\bu_j^\top\widetilde{\bu}_j\ge0\). Then
\begin{align}
 |\widetilde\lambda_j-\lambda_j-\bu_j^\top\bE\bu_j|
 &\le C\|\bE\|_{\op}^2/g_j,\label{supp:eq:eigenvalue-perturbation}\\
 \|\widetilde{\bu}_j-\bu_j-\bG_j\bE\bu_j\|_2
 &\le C\|\bE\|_{\op}^2/g_j^2.
 \label{supp:eq:eigenvector-perturbation}
\end{align}
The constants are independent of the dimension.
\end{lemma}

For fixed \(j\le m\), write
\[
 \bq_j(\bx)
 =\frac{d+2}{d}\bG_j\bPsi(\bx)\bu_j
 =\frac{d+2}{d}\sum_{k\ne j}\eta_{jk}(\bx)\bu_k,
 \qquad
 \bPi_{\mathrm{bulk}}=\sum_{k>m}\bu_k\bu_k^\top.
\]
\begin{lemma}[Spectral moments and bulk concentration]
\label{supp:lem:moments}
The eigenvalue score \(\varphi_{\lambda,j}\) in Theorem~\ref{thm:one-sample-eigenvalue}
and the eigendirection scores in \eqref{eq:eigenvector-influence-score} have mean zero and satisfy
\[
 \Var\left\{\frac{d+2}{d}\varphi_{\lambda,j}(\bx)\right\}
 =\sigma_{\lambda,j}^2,\qquad
 \Cov\{\bu_k^\top\bq_j(\bx),
       \bu_l^\top\bq_j(\bx)\}
 =\sigma_{u,jk}^2\ind(k=l)\quad(k,l\ne j).
\]
Here the variances are those in
\eqref{eq:eigenvalue-clt-main} and \eqref{eq:direction-population-weight}.
The eigenvalue variance has a positive limit, and
\[
 \E|\varphi_{\lambda,j}(\bx)|^4\le C,\qquad
 \E\|\bq_j(\bx)\|_2^4\le C.
\]
Furthermore,
\begin{equation}\label{supp:eq:bulk-concentration}
 \left\|\frac1{\sqrt n}\sum_{i=1}^n
       \bPi_{\mathrm{bulk}}\bq_j(\bx_i)\right\|_2^2-b_{u,j}
 =O_{\Pp}(d^{-1/2}+n^{-1/2}).
\end{equation}
\end{lemma}

We now prove the results in their main-text order.

\begin{proof}[Proof of Theorem~\ref{thm:common-expansion}]
First consider trace normalization. Set
\[
 \bH_n=\bM_n-\bSigma,\qquad
 \upsilon_n=d^{-1}\tr \bH_n,\qquad \bD_n=\widetilde{\bSigma}-\bSigma,
\]
and define the linear map
\(\mathcal L(\bH)=\bH-d^{-1}\tr(\bH)\bSigma\).
Since \(\|\bSigma\|_{\mathrm F}\le d\),
\(\|\mathcal L(\bH)\|_\diamond\le C\|\bH\|_\diamond\).
Lemma~\ref{supp:lem:raw-linear} gives
\(\|\bH_n\|_\diamond+|\upsilon_n|=O_{\Pp}(\epsilon_n)\).
The exact normalization identity is
\[
 \bD_n=\frac{\mathcal L(\bH_n)}{1+\upsilon_n}.
\]
In particular, \(\tr \bD_n=0\) and
\(d^{-1}\|\bD_n\|_{\mathrm F}=O_{\Pp}(\epsilon_n)\).
Applying \(\mathcal L\) to \eqref{supp:eq:raw-matrix-expansion} yields
\begin{equation}\label{supp:eq:normalized-matrix-expansion}
 \bD_n=\frac1n\sum_{i=1}^n\bPsi(\bx_i)-\frac{2}{d+2}\bM_0
     +\bR_{N,n},\qquad
 \|\bR_{N,n}\|_\diamond=O_{\Pp}(\rho_n),
\end{equation}
where
\[
 \bM_0=\bSigma\bar{\bDelta}_P\bSigma
   -\frac{\tr(\bSigma^2\bar{\bDelta}_P)}d\bSigma.
\]
Indeed, \(\mathcal L(\bPhi)=\bPsi\), the scalar multiple of
\(\bSigma\) in the raw derivative disappears, and replacing
\((1+\upsilon_n)^{-1}\) by one contributes
\(O_{\Pp}(\epsilon_n^2)\) in \(\|\cdot\|_\diamond\).
This also proves the linear-pilot comparison in
Remark~\ref{rem:neyman-orthogonality}.

Next consider the calibrated correction. The exact scale relation is
\[
 \widehat{c}_s=c_d(1+\upsilon_n),\qquad
 \widetilde{\bOmega}
 =(1+\upsilon_n)(\bOmega+\bar{\bDelta}_P).
\]
For \(\mathcal F(\bA,\bP)=\bA\bP\bA-d^{-1}\tr(\bA^2\bP)\bA\), put
\[
 \bQ_n=\bD_n\bOmega \bD_n+\bSigma\bar{\bDelta}_P\bD_n
       +\bD_n\bar{\bDelta}_P\bSigma
       +\bD_n\bar{\bDelta}_P\bD_n.
\]
Expanding the products and using \(\tr \bD_n=0\) gives
\[
 \mathcal F(\widetilde{\bSigma},\bOmega+\bar{\bDelta}_P)
 =\bD_n+\bM_0+\mathcal R_n,\qquad
 \mathcal R_n=\bQ_n-\frac{\tr \bQ_n}{d}\widetilde{\bSigma}
             -\frac{\tr(\bSigma^2\bar{\bDelta}_P)}d\bD_n.
\]
The bounds
\(\|\bD_n\|_{\mathrm F}=O_{\Pp}(d\epsilon_n)\),
\(\|\bar{\bDelta}_P\|_{\op}=O_{\Pp}(a_n)\), and
\(\|\bOmega\|_{\op}\le C\) imply
\begin{align*}
 \|\bQ_n\|_{\mathrm F}+|\tr \bQ_n|
 &\le C\|\bD_n\|_{\mathrm F}^2
   +C\|\bar{\bDelta}_P\|_{\op}
      \{d\|\bD_n\|_{\mathrm F}+\|\bD_n\|_{\mathrm F}^2\}
   =O_{\Pp}(d^2\rho_n),\\
 |\tr(\bSigma^2\bar{\bDelta}_P)|
 &\le\|\bar{\bDelta}_P\|_{\op}\tr(\bSigma^2)
   =O_{\Pp}(d^2a_n).
\end{align*}
Consequently,
\[
 \|\bM_0\|_{\mathrm F}=O_{\Pp}(d^2a_n),\qquad
 \|\mathcal R_n\|_{\mathrm F}=O_{\Pp}(d^2\rho_n).
\]
Linearity of \(\mathcal F\) in its second argument now gives
\[
 \mathcal F(\widetilde{\bSigma},\widetilde{\bOmega})
 =\bD_n+\bM_0+\widetilde{\mathcal R}_n,\qquad
 \widetilde{\mathcal R}_n
 =\mathcal R_n+\upsilon_n(\bD_n+\bM_0+\mathcal R_n),
 \qquad
 d^{-2}\|\widetilde{\mathcal R}_n\|_{\mathrm F}=O_{\Pp}(\rho_n).
\]
It follows that
\[
 \widehat{\bSigma}-\bSigma
 =\left(1+\frac2d\right)\bD_n+\frac2d\bM_0
  +\frac2d\widetilde{\mathcal R}_n.
\]
Substitution of \eqref{supp:eq:normalized-matrix-expansion} cancels
the entire linear precision term because
\[
 -\left(1+\frac2d\right)\frac2{d+2}+\frac2d=0.
\]
Thus \eqref{eq:common-matrix-expansion} holds with
\begin{equation}\label{supp:eq:common-frobenius-remainder}
 \bR_n=\left(1+\frac2d\right)\bR_{N,n}
       +\frac2d\widetilde{\mathcal R}_n,\qquad
 d^{-1}\|\bR_n\|_{\mathrm F}=O_{\Pp}(\rho_n).
\end{equation}
This joint matrix bound implies the required spectral bounds:
\[
 \frac{|\bu_j^\top\bR_n\bu_j|}{\lambda_j}
 +\|\bG_j\bR_n\bu_j\|_2
 \le(\lambda_j^{-1}+g_j^{-1})\|\bR_n\|_{\mathrm F}
 =O_{\Pp}(\rho_n).
\]
Finally,
\(\E\bPsi=0\) and
\(\E\|\bPsi\|_{\mathrm F}^2\le Cd^2\)
by Lemma~\ref{supp:lem:angular-denominators}.
Therefore the oracle average has Frobenius norm
\(O_{\Pp}(dn^{-1/2})\).
Together with \eqref{supp:eq:common-frobenius-remainder}, this proves
\eqref{eq:common-matrix-rate}, without using any spectral conclusion.
\end{proof}

\begin{proof}[Proof of Theorem~\ref{thm:one-sample-eigenvalue}]
Write \(\bE_n=\widehat{\bSigma}-\bSigma\).
Theorem~\ref{thm:common-expansion} gives
\(\|\bE_n\|_{\op}/g_j=O_{\Pp}(\epsilon_n)=o_{\Pp}(1)\).
Hence the event \(\|\bE_n\|_{\op}\le g_j/4\) has probability tending to
one. On this event, Lemma~\ref{supp:lem:spectral-perturbation} gives
\[
 \frac{\widehat\lambda_j-\lambda_j}{\lambda_j}
 =\frac{\bu_j^\top\bE_n\bu_j}{\lambda_j}
  +O_{\Pp}(\epsilon_n^2)
 =\frac{d+2}{d}\frac1n\sum_{i=1}^n
     \varphi_{\lambda,j}(\bx_i)+O_{\Pp}(\rho_n).
\]
This proves \eqref{eq:eigenvalue-expansion-main}.
By Lemma~\ref{supp:lem:moments}, the summands are centered, have
uniformly bounded fourth moments, and have variance
\(\{d/(d+2)\}^2\sigma_{\lambda,j}^2\), bounded away from zero.
The Lyapunov ratio is therefore \(O(n^{-1})\).
The triangular-array central limit theorem and
\(\sqrt n\,\rho_n\to0\) prove \eqref{eq:eigenvalue-clt-main}.
The high-probability restriction above does not change this limit.

For later use, the same expansion gives
\begin{equation}\label{supp:eq:log-eigenvalue-expansion}
 \log\widehat\lambda_j-\log\lambda_j
 =\frac{d+2}{d}\frac1n\sum_{i=1}^n
       \varphi_{\lambda,j}(\bx_i)+O_{\Pp}(\rho_n).
\end{equation}
Indeed, the relative eigenvalue error is \(O_{\Pp}(n^{-1/2})\), and
the quadratic logarithmic remainder is \(O_{\Pp}(n^{-1})\).
The event that the leading fitted eigenvalues are positive has probability
tending to one; arbitrary fixed extensions outside this event are immaterial.
\end{proof}

\begin{proof}[Proof of Theorem~\ref{thm:one-sample-eigenvector}]
Write \(\bE_n=\widehat{\bSigma}-\bSigma\)
and align the fitted eigenvector with \(\bu_j\).
On the same spectral event as above,
Lemma~\ref{supp:lem:spectral-perturbation} and
Theorem~\ref{thm:common-expansion} give
\begin{equation}\label{supp:eq:eigenvector-linear-representation}
 \widehat{\bu}_j-\bu_j
 =\bG_j\bE_n\bu_j+O_{\Pp}(\rho_n)
 =\frac1n\sum_{i=1}^n\bq_j(\bx_i)
   +\br_{n,j},\qquad
 \|\br_{n,j}\|_2=O_{\Pp}(\rho_n).
\end{equation}
Here and below, a vector \(O_{\Pp}\) bound is measured in the Euclidean norm.
This includes the component parallel to \(\bu_j\), so it proves the
vector expansion \eqref{eq:one-sample-eigenvector-expansion}.

For the finitely many \(k\le m\), \(k\ne j\), apply the
multivariate Lyapunov theorem to
\((\bu_k^\top\bq_j(\bx_i))_{k\le m,k\ne j}\).
Lemma~\ref{supp:lem:moments} gives bounded fourth moments and a diagonal
covariance matrix with entries \(\sigma_{u,jk}^2\).
The spike limits imply
\[
 \sigma_{u,jk}^2\longrightarrow
 \frac{\kappa_j\kappa_k}{(\kappa_j-\kappa_k)^2}
 =\sigma_{u,jk}^{*2}>0.
\]
Combining this central limit theorem with
\(\sqrt n\,\|\br_{n,j}\|_2=o_{\Pp}(1)\) gives the intermediate limit
\begin{equation}\label{supp:eq:leading-direction-clt}
 \left(\sqrt n\,\bu_k^\top(\widehat{\bu}_j-\bu_j)\right)_{
  \substack{1\le k\le m\\k\ne j}}
 \xrightarrow{\mathcal D}
 N\!\left(0,\Diag\{\sigma_{u,jk}^{*2}:k\le m,\ k\ne j\}\right).
\end{equation}

To derive the limit for the squared sine of the angle between
\(\widehat{\bu}_j\) and \(\bu_j\), let
\(\bP_j^\perp=\bI_d-\bu_j\bu_j^\top\) and
\(\bt_{n,j}=n^{-1/2}\sum_i\bq_j(\bx_i)\).
Lemma~\ref{supp:lem:moments} gives
\(\E\|\bt_{n,j}\|_2^2
=\sum_{k\ne j}\sigma_{u,jk}^2=O(1)\).
Both \(\bq_j\) and \(\bt_{n,j}\) are orthogonal
to \(\bu_j\). Since the fitted eigenvector has unit norm,
\[
 n\left[1-\left|\bu_j^\top\widehat{\bu}_j\right|^2\right]
 =\|\bt_{n,j}
       +\sqrt n\,\bP_j^\perp\br_{n,j}\|_2^2
 =\|\bt_{n,j}\|_2^2+o_{\Pp}(1).
\]
The last error is
\(O_{\Pp}(\sqrt n\,\rho_n+n\rho_n^2)=o_{\Pp}(1)\).
The leading and bulk projections are exactly orthogonal, so
\begin{equation}\label{supp:eq:direction-leading-bulk}
 \begin{aligned}
 n\left[1-\left|\bu_j^\top\widehat{\bu}_j\right|^2\right]-b_{u,j}
 ={}&\sum_{\substack{k\le m\\k\ne j}}
        (\bu_k^\top\bt_{n,j})^2\\
 &+O_{\Pp}(\sqrt n\,\rho_n+d^{-1/2}+n^{-1/2}).
 \end{aligned}
\end{equation}
Here the bulk bound is \eqref{supp:eq:bulk-concentration}.
The joint central limit theorem in
\eqref{supp:eq:leading-direction-clt} gives the Gaussian limit
\((\bD_j^*)^{1/2}\bz_j\) for the leading error vector.
Its squared norm is \(\bz_j^\top\bD_j^*\bz_j
=\sum_{k\le m,\,k\ne j}\sigma_{u,jk}^{*2}Z_k^2\), proving
\eqref{eq:one-sample-full-eigenvector-limit}. If \(m=1\), the leading
sum is empty and the same argument yields convergence to zero.

Finally, for \(k>m\), \(\lambda_k/\lambda_j=O(d^{-1})\) uniformly.
Using \(|(1-x)^{-2}-1|\le C|x|\) for \(|x|\le1/2\), we obtain
\[
 \left|b_{u,j}
  -\frac{d+2}{d}\frac{d-\sum_{a=1}^m\lambda_a}{\lambda_j}\right|
 \le\frac{C}{\lambda_j^2}\sum_{k>m}\lambda_k^2
 =O(d^{-1}).
\]
This proves the trace approximation for \(b_{u,j}\) stated after
Theorem~\ref{thm:one-sample-eigenvector}.
Also, Theorem~\ref{thm:one-sample-eigenvalue} implies
\(\max_{a\le m}|\widehat\lambda_a/d-\lambda_a/d|
=O_{\Pp}(n^{-1/2})\).
The functions
\[
 (x,y)\longmapsto\frac{d+2}{d}\frac{xy}{(x-y)^2},
 \qquad
 (x_1,\ldots,x_m)\longmapsto
      \frac{d+2}{d}\frac{1-\sum_{a=1}^m x_a}{x_j}
\]
have bounded first derivatives near the distinct positive spike
limits. Thus their plug-in versions, defined as in
Section~\ref{subsec:two-sample-eigenvectors} without a group index, satisfy
\begin{equation}\label{supp:eq:direction-plugin-rates}
 \widehat\sigma_{u,jk}^2-\sigma_{u,jk}^2
 =O_{\Pp}(n^{-1/2})\quad(k\le m,\ k\ne j),\qquad
 \widehat b_{u,j}-b_{u,j}
 =O_{\Pp}(n^{-1/2}+d^{-1}).
\end{equation}
These conclusions also justify the one-sample centering and weight
estimates used in Section~\ref{subsec:two-sample-eigenvectors}.
\end{proof}

\begin{proof}[Proof of Proposition~\ref{prop:nuisance-construction}]
This verification uses Assumptions~\ref{ass:factor-structure} and
\ref{ass:inverse-radial-moments}, together with the additional
conditions stated in the proposition; it does not assume
Assumption~\ref{ass:nuisance-estimation}.
Since \(s_d\ge c_{\varepsilon}^v>0\) and
\[
 \delta(n)\asymp\sqrt{\frac{\log(dn)}n}+d^{-1/2},
\]
the growth conditions in Proposition~\ref{prop:nuisance-construction} give
\[
 n^{1/2}s_d^2\{\delta(n)\}^{2(1-v)}
 \asymp
 \frac{s_d^2\{\log(dn)\}^{1-v}}{n^{1/2-v}}
 +\frac{s_d^2n^{1/2}}{d^{1-v}}
 \longrightarrow0.
\]
Moreover, \(d^{-1}=O\{s_d d^{-(1-v)/2}\}\), so
\(a_n=o(n^{-1/4})\).
These conditions also imply \(\log d=o(n)\), \(\log n=o(d)\),
and \(s_d\{\delta(n)\}^{1-v}=o(1)\), as required for the POET--SS bounds below.

We first verify the location rate separately. This step only needs
Assumption~\ref{ass:factor-structure} and
\(\E(R^{-2})/\zeta_1^2\le C\), where \(R=\|\bx-\bmu\|_2\)
and \(0<\zeta_1=\E(R^{-1})<\infty\).
The latter bound follows from the \(q=2\) case of
Assumption~\ref{ass:inverse-radial-moments}, since
\[
 \frac{\E(R^{-2})}{\zeta_1^2}
 =\frac{\E(\xi^{-2})}{\{\E(\xi^{-1})\}^2}
   \frac{\E(\mathcal A^{-1})}{\{\E(\mathcal A^{-1/2})\}^2}
 \le c_{\varepsilon}^{-1}\frac{\E(\xi^{-2})}{\{\E(\xi^{-1})\}^2}
 \le C,
 \qquad \mathcal A=\bs^\top\bSigma\bs.
\]
For a training sample of size \(N\), suppress the fold index and put
\(R_i=\|\bx_i-\bmu\|_2\), \(\bt_i=\spSign(\bx_i-\bmu)\), and
\(\alpha_i=(\zeta_1R_i)^{-1}\). Define
\[
 \bA_N=\frac1N\sum_{i=1}^N\alpha_i(\bI_d-\bt_i\bt_i^\top),
 \qquad \bA=\E\bA_N.
\]
The population curvature satisfies \(\lambda_{\min}(\bA)\ge c>0\)
uniformly in \(d\). To verify this, for a unit vector \(\bv\), let
\(\mathcal B_{\bv}=\mathcal A-(\bv^\top\bSigma^{1/2}\bs)^2\).
Radial independence gives
\[
 \bv^\top\bA\bv
 =\frac{\E(\mathcal B_{\bv}\mathcal A^{-3/2})}
        {\E(\mathcal A^{-1/2})}.
\]
Here \(0\le\mathcal B_{\bv}\le\mathcal A\),
\(\E\mathcal B_{\bv}=1-\bv^\top\bSigma\bv/d\ge c_{\varepsilon}(1-d^{-1})\),
and the spherical moment formula gives \(\E\mathcal A^2\le3\).
For a sufficiently large fixed \(M_0\) and \(d\ge2\),
\[
 \E\{\mathcal B_{\bv}\ind(\mathcal A\le M_0)\}
 \ge c_{\varepsilon}/2-3/M_0\ge c_{\varepsilon}/4.
\]
Together with \(\E(\mathcal A^{-1/2})\le c_{\varepsilon}^{-1/2}\), this proves
the curvature bound.
Since \(\E\alpha_i=1\), \(\E\alpha_i^2\le C\), and
\(\|\bt_i\bt_i^\top\|_{\mathrm F}=1\), independence yields
\[
 \|\bA_N-\bA\|_{\op}
 \le\left|\frac1N\sum_i\alpha_i-1\right|
 +\left\|\frac1N\sum_i
   \{\alpha_i\bt_i\bt_i^\top-\E(\alpha_i\bt_i\bt_i^\top)\}
   \right\|_{\mathrm F}
 =O_{\Pp}(N^{-1/2}).
\]
Also, \(\bar{\bt}=N^{-1}\sum_i\bt_i\) satisfies
\(\E\|\bar{\bt}\|_2^2=N^{-1}\) by symmetry.

For the standardized displacement \(\bh\), consider the convex objective
\[
 L_N(\bh)=\frac{\zeta_1}N\sum_{i=1}^N
 \{\|\bx_i-\bmu-\bh/\zeta_1\|_2-R_i\}.
\]
The Euclidean norm admits the global remainder bound
\[
 \left|L_N(\bh)+\bar{\bt}^{\top}\bh
       -\tfrac12\bh^\top\bA_N\bh\right|
 \le C\|\bh\|_2^3\frac1N\sum_i\alpha_i^2.
\]
For each summand, this follows from Taylor expansion when
\(\|\bh\|_2/\zeta_1\le R_i/2\); otherwise, the triangle inequality
bounds the constant, linear, and quadratic terms by the same cubic bound.
Thus no bound on \(\max_i R_i^{-1}\) is needed.
On the sphere \(\|\bh\|_2=M/\sqrt N\), the preceding estimates give
\[
 \inf_{\|\bh\|_2=M/\sqrt N}L_N(\bh)
 \ge\frac1N\{cM^2/4-MO_{\Pp}(1)-O_{\Pp}(M^3/\sqrt N)\}.
\]
For every \(\varepsilon>0\), choose \(M\) sufficiently large
and then \(N\) sufficiently large so that the right-hand side is positive
with probability at least \(1-\varepsilon\). Since \(L_N(\boldsymbol0)=0\),
convexity places the minimizer inside this sphere and proves
\(\zeta_1\|\widehat{\bmu}-\bmu\|_2=O_{\Pp}(N^{-1/2})\).
With \(N=N_{-\ell}\asymp n\) and fixed \(K\), this yields the
uniform location rate. This argument uses neither a restriction on the relative growth of \(n\) and
\(d\) nor a sparsity condition.

We next verify the precision-pilot rate. Here the stronger assumptions
in the proposition are used to control the spatial-sign matrix entrywise,
including the effect of estimating its center.
Let \(\bGamma_m\) contain the leading \(m\) orthonormal eigenvectors
of the scaled population spatial-sign covariance matrix \(\bS\), and let \(\bLambda_m\) contain
the corresponding eigenvalues on its diagonal. The POET--SS population
target is
\[
 \bH=\bGamma_m\bLambda_m\bGamma_m^\top+c_d\bSigma_{\varepsilon}.
\]
Proposition~2.1 and Theorems~2.1--2.2 of
\citet{supp:XuMaWangFeng2025}, applied with a training sample size of
\(N_{-\ell}\asymp n\) and hence \(\delta(N_{-\ell})\asymp\delta(n)\), give
\begin{align*}
 \|\widehat{\bH}_{\varepsilon,-\ell}-c_d\bSigma_{\varepsilon}\|_{\op}
 &=O_{\Pp}(s_d\{\delta(n)\}^{1-v}),\\
 \|(\widehat{\bH}_{-\ell})^{-1}
       -\bH^{-1}\|_{\op}
 &=O_{\Pp}(s_d\{\delta(n)\}^{1-v}),\qquad
 \|c_d\bH^{-1}-\bOmega\|_{\op}=O(d^{-1}),
\end{align*}
where the inverse bound holds on an event of probability tending to
one. The population scalar
\(c_d=d\E(Q^{-1})\) depends only on \(\bSigma\) and hence is common
to all folds. Since \(\E Q=d\) and
\(Q\ge c_{\varepsilon}\|\bz\|_2^2\), Jensen's inequality and
\(\E(\chi_d^2)^{-1}=1/(d-2)\) give
\[
 1\le c_d\le\frac{d}{c_{\varepsilon}(d-2)}\le C_c
\]
for all sufficiently large \(d\).

The retained leading component of
\(\widehat{\bH}_{-\ell}\) is positive
semidefinite. Consequently, the minimum eigenvalue of
\(\widehat{\bH}_{-\ell}\) is at least
\[
 \lambda_{\min}(\widehat{\bH}_{\varepsilon,-\ell})
 \ge c_dc_{\varepsilon}-
   \|\widehat{\bH}_{\varepsilon,-\ell}-c_d\bSigma_{\varepsilon}\|_{\op}
 \ge c_{\varepsilon}/2
\]
with probability tending to one. As \(d^{-1}<c_{\varepsilon}/2\) eventually,
the diagonal adjustment is zero on this event. The triangle
inequality therefore gives
\[
 \|c_d\widehat{\bP}_{-\ell}-\bOmega\|_{\op}
 \le c_d\|\widehat{\bP}_{-\ell}-\bH^{-1}\|_{\op}
       +\|c_d\bH^{-1}-\bOmega\|_{\op}
 =O_{\Pp}(s_d\{\delta(n)\}^{1-v}+d^{-1}).
\]
Since \(K\) is fixed, all these statements hold uniformly over the folds.
Positive definiteness and training measurability hold by construction.
Moreover, \(s_d\ge c_{\varepsilon}^v\) and \(\delta(n)\to0\) imply
\(a_n\ge c_{\varepsilon}^v\delta(n)\ge c_{\varepsilon}^v n^{-1/2}\) eventually, so the
standardized spatial-median error is also \(O_{\Pp}(a_n)\).
Together with \(a_n=o(n^{-1/4})\), these bounds establish
Assumption~\ref{ass:nuisance-estimation}.
\end{proof}

\subsection{Proofs for section~\ref*{subsec:two-sample-eigenvalues}}
\label{supp:subsec:section-4-1-proofs}

\begin{proof}[Proof of Theorem~\ref{thm:two-sample-eigenvalue-test}]
First, we verify variance consistency. Let \(\widehat{\bSigma}^{(r)}\)
denote the orthogonalized shape estimator for group \(r\), put
\(\bE^{(r)}=\widehat{\bSigma}^{(r)}-\bSigma^{(r)}\), and write
\(\epsilon^{(r)}\) for the group-specific version of \(\epsilon_n\).
Theorem~\ref{thm:common-expansion} gives
\(d^{-1}\|\bE^{(r)}\|_{\mathrm F}=O_{\Pp}(\epsilon^{(r)})\), where
\(\epsilon^{(r)}\to0\). Since \(\|\bSigma^{(r)}\|_{\mathrm F}\le d\),
\begin{equation}\label{supp:eq:eigenvalue-square-sum-consistency}
 \begin{aligned}
 \frac1{d^2}\left|\sum_{k=1}^d
  \{(\widehat\lambda_k^{(r)})^2-(\lambda_k^{(r)})^{2}\}\right|
 &=\frac1{d^2}\left|2\tr(\bSigma^{(r)}\bE^{(r)})+\tr((\bE^{(r)})^{2})\right|\\
 &\le\frac{2\|\bSigma^{(r)}\|_{\mathrm F}\|\bE^{(r)}\|_{\mathrm F}
             +\|\bE^{(r)}\|_{\mathrm F}^2}{d^2}
 =O_{\Pp}(\epsilon^{(r)}).
 \end{aligned}
\end{equation}
This controls the full spectral sum without requiring consistency of
individual bulk eigenvalues. Also,
Theorem~\ref{thm:one-sample-eigenvalue} and \(\lambda_j^{(r)}\asymp d\)
give
\(d^{-1}|\widehat\lambda_j^{(r)}-\lambda_j^{(r)}|
=O_{\Pp}(n_r^{-1/2})\).
Substituting these bounds into
\eqref{eq:eigenvalue-variance-estimator} yields
\begin{equation}\label{supp:eq:eigenvalue-variance-consistency}
 (\widehat\sigma_{\lambda,j}^{(r)})^{2}-(\sigma_{\lambda,j}^{(r)})^{2}
 =O_{\Pp}(\epsilon^{(r)})=o_{\Pp}(1),\qquad r=1,2.
\end{equation}
By Lemma~\ref{supp:lem:moments}, applied separately with factor ranks
\(m_1\) and \(m_2\), there are constants \(0<c<C<\infty\) such that
\(c\le(\sigma_{\lambda,j}^{(r)})^{2}\le C\) for both groups and all
sufficiently large \(d\). This proves the claimed ratio consistency.
In particular, the fitted variances and the leading fitted eigenvalues
are positive with probability tending to one.

Next, consider the two-sample log contrast. Write
\(\varphi_{\lambda,j}^{(r)}\) for the group-specific version of the score in
Theorem~\ref{thm:one-sample-eigenvalue} and put \(\rho^{(r)}=(\epsilon^{(r)})^{2}\).
Equation~\eqref{supp:eq:log-eigenvalue-expansion} gives
\[
 \log\widehat\lambda_j^{(r)}-\log\lambda_j^{(r)}
 =\frac{\sigma_{\lambda,j}^{(r)}}{\sqrt{n_r}}\mathcal Z_n^{(r)}
   +O_{\Pp}(\rho^{(r)}),\qquad
 \mathcal Z_n^{(r)}
 =\frac{d+2}{d\sigma_{\lambda,j}^{(r)}\sqrt{n_r}}
   \sum_{i=1}^{n_r}\varphi_{\lambda,j}^{(r)}(\bx_i^{(r)}).
\]
Here \(\sqrt{n_r}\rho^{(r)}\to0\), and the logarithm is considered on
the event of a positive fitted eigenvalue. The one-sample central
limit theorem applies to each oracle sum. These sums depend only on
their respective independent samples, so
\((\mathcal Z_n^{(1)},\mathcal Z_n^{(2)})
\xrightarrow{\mathcal D}N(0,\bI_2)\).
Put
\[
 s_{j,n}^2=\frac{(\sigma_{\lambda,j}^{(1)})^{2}}{n_1}
           +\frac{(\sigma_{\lambda,j}^{(2)})^{2}}{n_2},\qquad
 \alpha_n^{(r)}=\frac{\sigma_{\lambda,j}^{(r)}}{\sqrt{n_r}\,s_{j,n}}.
\]
Under \(\mathbb{H}_{0,\lambda,j}\), on the event that both fitted eigenvalues
are positive,
\[
 \frac{\log\widehat\lambda_j^{(1)}
       -\log\widehat\lambda_j^{(2)}}{s_{j,n}}
 =\alpha_n^{(1)}\mathcal Z_n^{(1)}-\alpha_n^{(2)}\mathcal Z_n^{(2)}
   +o_{\Pp}(1).
\]
Indeed, \(s_{j,n}\ge c^{1/2}n_r^{-1/2}\) for each group, so
\((\rho^{(1)}+\rho^{(2)})/s_{j,n}
\le C\sum_{r=1}^2\sqrt{n_r}\rho^{(r)}\to0\).
The deterministic weights satisfy
\((\alpha_n^{(1)})^{2}+(\alpha_n^{(2)})^{2}=1\). Along any subsequence, one can
extract a further subsequence on which these weights converge to
\((\alpha^{(1)},\alpha^{(2)})\); the resulting limit is
\(\alpha^{(1)}Z^{(1)}-\alpha^{(2)}Z^{(2)}\sim N(0,1)\) for independent standard
normal \(Z^{(1)},Z^{(2)}\). Thus the population-standardized log contrast
converges to \(N(0,1)\) along the full sequence. This argument requires
no restriction on the relative growth of \(n_1\) and \(n_2\).

Finally, set
\(\widehat s_{j,n}^2=(\widehat\sigma_{\lambda,j}^{(1)})^{2}/n_1
                       +(\widehat\sigma_{\lambda,j}^{(2)})^{2}/n_2\).
The ratio consistency already proved implies
\[
 \left|\frac{\widehat s_{j,n}^2}{s_{j,n}^2}-1\right|
 \le\max_{r=1,2}
   \left|\frac{(\widehat\sigma_{\lambda,j}^{(r)})^{2}}
                  {(\sigma_{\lambda,j}^{(r)})^{2}}-1\right|
 =o_{\Pp}(1).
\]
Slutsky's theorem proves \eqref{eq:two-sample-eigenvalue-null}.
The convention \(Z_{\lambda,j}=0\) outside the positivity event does
not affect this limit because that event has probability tending to
one. Continuity of the standard-normal distribution then gives
\[
 \Pp_{\mathbb{H}_{0,\lambda,j}}\{|Z_{\lambda,j}|>z_{1-\alpha/2}\}
 \longrightarrow\alpha,\qquad
 p_{\lambda,j}\xrightarrow{\mathcal D}\operatorname{Unif}(0,1).
\]
All arguments use the group-specific one-sample results and sample
independence, without equating the radial distributions or the
eigendirections of the two populations.
\end{proof}

\subsection{Proofs for section~\ref*{subsec:two-sample-eigenvectors}}
\label{supp:subsec:section-4-2-proofs}
We use the following elementary bound to compare Gaussian quadratic
forms, including those with singular coefficient matrices. Its proof
is given in Section~\ref{supp:subsec:auxiliary-lemma-proofs}.

\begin{lemma}[Continuity of Gaussian quadratic-form laws]
\label{supp:lem:quadratic-cdf}
Let \(\bK,\bL\) be positive semidefinite \(q\times q\) matrices, where
\(q\ge1\) is fixed, and suppose \(\lambda_{\max}(\bK)\ge c>0\).
For \(\bg\sim N(0,\bI_q)\), write \(F_{\bK}(x)=\Pp(\bg^\top \bK\bg\le x)\),
and define \(F_{\bL}\) similarly. Then \(F_{\bK}\) is continuous, is strictly
increasing on \([0,\infty)\), and satisfies \(F_{\bK}(0)=0\). Moreover, for
\(t>0\),
\begin{equation}\label{supp:eq:quadratic-cdf-bounds}
 \sup_x\{F_{\bK}(x+t)-F_{\bK}(x)\}\le C\sqrt t,\qquad
 \sup_x|F_{\bK}(x)-F_{\bL}(x)|\le C\|\bK-\bL\|_{\op}^{1/3},
\end{equation}
where \(C\) depends only on \(c\) and \(q\).
\end{lemma}

\begin{proof}[Proof of Theorem~\ref{thm:two-sample-eigendirection-test}]
Under \(\mathbb{H}_{0,u,j}\), choose the population signs so that
\(\bu_j^{(1)}=\bu_j^{(2)}=\bu_j\). For the proof, also align each fitted
target eigenvector with \(\bu_j\); the squared inner product is
unchanged. Write \(w^{(r)}=n_{12}/n_r\), and let \(\epsilon^{(r)},\rho^{(r)}\)
be the group-specific versions of \(\epsilon_n,\rho_n\).
Empty leading blocks when \(m_r=1\) are interpreted as having zero
columns. Let \(\mathcal G_n\) be the event that all fitted leading
eigenvalues in both groups are positive and simple.
Theorem~\ref{thm:common-expansion} and the population
eigengaps give \(\Pp(\mathcal G_n)\to1\).
All fitted reference matrices below are considered on this event.
For definiteness, on its complement set \(T_{u,j}=0\) and take
\(\widehat F_{u,j}\) to be the \(\chi_{m_u}^2\) distribution
function, consistently with the convention \(p_{u,j}=1\).
These extensions do not affect the asserted limits.

First, we reduce the angular discrepancy to a quadratic form in the
oracle errors. Define the group-specific scores and their sums by
\[
 \bq_j^{(r)}(\bx)
 =\frac{d+2}{d}\sum_{k\ne j}\eta_{jk}^{(r)}(\bx)\bu_k^{(r)},\qquad
 \bt_j^{(r)}
 =\frac1{\sqrt{n_r}}\sum_{i=1}^{n_r}
       \bq_j^{(r)}(\bx_i^{(r)}).
\]
Theorem~\ref{thm:one-sample-eigenvector} gives
\[
 \widehat{\bu}_j^{(r)}-\bu_j
 =n_r^{-1/2}\bt_j^{(r)}+\br_j^{(r)},\qquad
 \|\br_j^{(r)}\|_2=O_{\Pp}(\rho^{(r)}),\qquad
 \sqrt{n_r}\rho^{(r)}\to0.
\]
By Lemma~\ref{supp:lem:moments},
\(\E\|\bt_j^{(r)}\|_2^2=O(1)\), so the eigenvector
error in Euclidean norm is \(O_{\Pp}(n_r^{-1/2})\).
Put \(\bP_j^\perp=\bI_d-\bu_j\bu_j^\top\) and write
\(\widehat{\bu}_j^{(r)}=\alpha^{(r)}\bu_j+\bh^{(r)}\),
where \(\bh^{(r)}=\bP_j^\perp\widehat{\bu}_j^{(r)}\) and
\(\alpha^{(r)}\ge0\). Unit normalization implies
\(1-\alpha^{(r)}=O_{\Pp}(n_r^{-1})\).
For \(t=\langle\widehat{\bu}_j^{(1)},\widehat{\bu}_j^{(2)}\rangle\),
\[
 2(1-t)=\|\bh^{(1)}-\bh^{(2)}\|_2^2+(\alpha^{(1)}-\alpha^{(2)})^2,
 \qquad 1-t=O_{\Pp}(n_{12}^{-1}).
\]
Using \(1-t^2=2(1-t)-(1-t)^2\) and the vector expansion above gives
\begin{equation}\label{supp:eq:two-sample-angular-reduction}
 \begin{aligned}
 n_{12}(1-t^2)
 &=n_{12}\|\bh^{(1)}-\bh^{(2)}\|_2^2+O_{\Pp}(n_{12}^{-1})\\
 &=\|\sqrt{w^{(1)}}\bt_j^{(1)}
       -\sqrt{w^{(2)}}\bt_j^{(2)}\|_2^2
   +O_{\Pp}\left(n_{12}^{-1}
                 +\sum_{r=1}^2\sqrt{n_r}\rho^{(r)}\right).
 \end{aligned}
\end{equation}
Here \(\bt_j^{(r)}\perp\bu_j\); the squared remainder is
absorbed because \(\sum_r\sqrt{n_r}\rho^{(r)}=o(1)\).

Next, split each oracle sum into its leading and bulk projections.
Let \(\bPi_{\mathrm{bulk}}^{(r)}=\sum_{k>m_r}\bu_k^{(r)}(\bu_k^{(r)})^{\top}\)
and set
\[
 \by^{(r)}=\sqrt{w^{(r)}}(\bGamma_{-j}^{(r)})^{\top}\bt_j^{(r)},
 \qquad
 \bl^{(r)}=\bGamma_{-j}^{(r)}\by^{(r)},\qquad
 \bb^{(r)}=\sqrt{w^{(r)}}\bPi_{\mathrm{bulk}}^{(r)}\bt_j^{(r)}.
\]
Thus \(\sqrt{w^{(r)}}\bt_j^{(r)}=\bl^{(r)}+\bb^{(r)}\)
and \(\langle\bl^{(r)},\bb^{(r)}\rangle=0\) exactly.
Lemma~\ref{supp:lem:moments} gives
\begin{align}
 \|\bb^{(r)}\|_2^2-w^{(r)}b_{u,j}^{(r)}
 &=O_{\Pp}(d^{-1/2}+n_r^{-1/2}),
 \label{supp:eq:two-sample-bulk-norm}\\
 \Cov(\bb^{(r)})
 &=w^{(r)}\sum_{k>m_r}(\sigma_{u,jk}^{(r)})^{2}\bu_k^{(r)}(\bu_k^{(r)})^{\top},
 \qquad \|\Cov(\bb^{(r)})\|_{\op}=O(d^{-1}).
 \label{supp:eq:two-sample-bulk-covariance}
\end{align}
All four vectors \(\bl^{(r)},\bb^{(r)}\) are centered
and have bounded second moments. For independent centered vectors
\(\bv,\bw\),
\[
 \E(\bv^\top\bw)^2
 =\tr\{\Cov(\bv)\Cov(\bw)\}
 \le\|\Cov(\bw)\|_{\op}\E\|\bv\|_2^2.
\]
Applying this identity across samples, with a bulk vector in the
role of \(\bw\), yields
\begin{equation}\label{supp:eq:two-sample-bulk-cross-products}
 \langle\bl^{(1)},\bb^{(2)}\rangle,
 \quad \langle\bb^{(1)},\bl^{(2)}\rangle,
 \quad \langle\bb^{(1)},\bb^{(2)}\rangle
 =O_{\Pp}(d^{-1/2}).
\end{equation}
This argument allows different bulk subspaces in the two groups and
does not require independence of the leading and bulk parts within
a group.

Let \(\bA_{j,d}\) and \(\bC_{j,d}\) be the population versions of
\(\widehat{\bA}_j\) and \(\widehat{\bC}_j\) in
\eqref{eq:plugin-direction-matrices}, obtained by replacing
\(\widehat{\boldsymbol\Theta}_{j,d}\) and
\((\widehat\sigma_{u,jk}^{(r)})^{2}\) by
\(\boldsymbol\Theta_{j,d}\) and \((\sigma_{u,jk}^{(r)})^{2}\).
Stack \(\by=((\by^{(1)})^{\top},
(\by^{(2)})^{\top})^\top\). Then
\[
 \|\bl^{(1)}-\bl^{(2)}\|_2^2
 =\by^\top\bA_{j,d}\by,
 \qquad \Cov(\by)=\bC_{j,d}.
\]
The plug-in centering bound in
\eqref{supp:eq:direction-plugin-rates}, together with
\eqref{supp:eq:two-sample-angular-reduction}--\eqref{supp:eq:two-sample-bulk-cross-products},
therefore proves
\begin{equation}\label{supp:eq:two-sample-direction-quadratic}
 T_{u,j}=\by^\top\bA_{j,d}\by
       +O_{\Pp}\left[
          d^{-1/2}+\sum_{r=1}^2
          \{\sqrt{n_r}\rho^{(r)}+n_r^{-1/2}\}\right]
 =\by^\top\bA_{j,d}\by+o_{\Pp}(1).
\end{equation}

We now identify the population reference law. The finite-dimensional
central limit argument in Theorem~\ref{thm:one-sample-eigenvector},
applied to the two independent oracle sums, gives
\[
 \by\xrightarrow{\mathcal D}N(0,\bC_j^*),\qquad
 \bC_{j,d}\longrightarrow\bC_j^*
 =\Diag\left(
   (\pi^{(1)}(\sigma_{u,jk}^{(1),*})^2)_{k\in\Kcal_j^{(1)}},
   (\pi^{(2)}(\sigma_{u,jk}^{(2),*})^2)_{k\in\Kcal_j^{(2)}}\right)\succ0.
\]
The uniformly bounded fourth moments in
Lemma~\ref{supp:lem:moments} verify Lyapunov's condition for every
fixed linear combination. Define
\[
 \bK_{u,j,d}=\bC_{j,d}^{1/2}\bA_{j,d}\bC_{j,d}^{1/2},\qquad
 Q_{j,d}=\bg^\top\bK_{u,j,d}\bg,\qquad
 F_{j,d}(x)=\Pp(Q_{j,d}\le x),
 \quad \bg\sim N(0,\bI_{m_u}).
\]
The subscript \(d\) also includes dependence on \(n_1,n_2\).
Since \(\bA_{j,d}\) is a Gram matrix,
\[
 \bA_{j,d}\succeq0,\qquad \|\bA_{j,d}\|_{\op}\le2,
 \qquad \tr(\bA_{j,d})=m_u.
\]
Its diagonal entries equal one, so
\(\tr(\bK_{u,j,d})=\tr(\bC_{j,d})\) is bounded away from zero.
Consequently, \(\lambda_{\max}(\bK_{u,j,d})\ge c>0\)
eventually, even when this matrix is singular.

The bound \(\|\boldsymbol\Theta_{j,d}\|_{\op}\le1\) places the
cross-Gram matrices in a compact set of fixed dimension. Every
subsequence therefore has a further subsequence on which
\(\boldsymbol\Theta_{j,d}\to\boldsymbol\Theta_j^*\) and
\(\bA_{j,d}\to\bA_j^*\). Along it,
\eqref{supp:eq:two-sample-direction-quadratic} and the central limit
theorem give
\[
 T_{u,j}\xrightarrow{\mathcal D}Q_j^*,\qquad
 Q_{j,d}\xrightarrow{\mathcal D}Q_j^*,\qquad
 Q_j^*=\bg^\top(\bC_j^*)^{1/2}\bA_j^*(\bC_j^*)^{1/2}\bg.
\]
The coefficient matrix of \(Q_j^*\) is nonzero and positive
semidefinite, so the distribution function of \(Q_j^*\) is continuous by
Lemma~\ref{supp:lem:quadratic-cdf}. When the limiting distribution function
is continuous, convergence in distribution implies uniform convergence
of the distribution functions.
Both distribution functions above thus converge uniformly to that
of \(Q_j^*\) along this further subsequence. Since the initial
subsequence was arbitrary, we obtain
\begin{equation}\label{supp:eq:population-direction-calibration}
 \sup_x|\Pp(T_{u,j}\le x)-F_{j,d}(x)|\longrightarrow0
\end{equation}
without a convergence assumption on \(\boldsymbol\Theta_{j,d}\).

Next, we verify that the fitted reference estimates this population
law. For each group, let \(\bD^{(r)}\) be a diagonal sign matrix aligning
the columns of \(\widehat{\bGamma}_{-j}^{(r)}\bD^{(r)}\)
with their population counterparts, and put \(\bD=\Diag(\bD^{(1)},\bD^{(2)})\).
The vector expansion in
Theorem~\ref{thm:one-sample-eigenvector}, applied to the finitely
many leading columns, gives
\[
 \|\widehat{\bGamma}_{-j}^{(r)}\bD^{(r)}
          -\bGamma_{-j}^{(r)}\|_{\op}=O_{\Pp}(n_r^{-1/2}).
\]
Together with \eqref{supp:eq:direction-plugin-rates}, this implies,
for \(\delta_{12,n}=n_1^{-1/2}+n_2^{-1/2}\),
\[
 \|\bD\widehat{\bA}_j\bD-\bA_{j,d}\|_{\op}
 =O_{\Pp}(\delta_{12,n}),\qquad
 \|\widehat{\bC}_j-\bC_{j,d}\|_{\op}
 =O_{\Pp}(\delta_{12,n}).
\]
The diagonal entries of \(\bC_{j,d}\) are bounded away from zero,
so the same rate holds for
\(\|\widehat{\bC}_j^{1/2}-\bC_{j,d}^{1/2}\|_{\op}\).
All matrices have bounded operator norm in probability. Since
\(\bD\) commutes with the diagonal matrix \(\widehat{\bC}_j\),
\begin{equation}\label{supp:eq:direction-reference-matrix-consistency}
 \|\bD\widehat{\bK}_{u,j}\bD-\bK_{u,j,d}\|_{\op}
 =O_{\Pp}(\delta_{12,n}).
\end{equation}
Conjugating by \(\bD\) leaves the eigenvalues and hence the fitted
reference law unchanged. Applying
Lemma~\ref{supp:lem:quadratic-cdf} conditionally on the data now yields
\begin{equation}\label{supp:eq:conditional-direction-calibration}
 \sup_x|\widehat F_{u,j}(x)-F_{j,d}(x)|
 =O_{\Pp}(\delta_{12,n}^{1/3})=o_{\Pp}(1).
\end{equation}
Combining this with
\eqref{supp:eq:population-direction-calibration} proves
\eqref{eq:two-sample-eigendirection-calibration}.

Finally, we justify the p-value and critical-value conclusions
following the theorem. By Lemma~\ref{supp:lem:quadratic-cdf},
\(F_{j,d}\) is continuous and strictly increasing on
\([0,\infty)\). For \(v\in(0,1)\), let
\(x_{d,v}=F_{j,d}^{-1}(v)>0\). Then
\[
 \left|\Pp\{F_{j,d}(T_{u,j})\le v\}-v\right|
 =|\Pp(T_{u,j}\le x_{d,v})-F_{j,d}(x_{d,v})|
 \longrightarrow0
\]
by \eqref{supp:eq:population-direction-calibration}.
On \(\mathcal G_n\), \(\widehat{\bK}_{u,j}\) is nonzero and
positive semidefinite, so its reference law is also continuous.
Consequently,
\[
 p_{u,j}=1-\widehat F_{u,j}(T_{u,j}),\qquad
 \left|p_{u,j}-\{1-F_{j,d}(T_{u,j})\}\right|
 \le\sup_x|\widehat F_{u,j}(x)-F_{j,d}(x)|=o_{\Pp}(1).
\]
This comparison does not require independence between the statistic
and its estimated reference distribution. The exceptional-event
convention is immaterial, and Slutsky's theorem gives
\(p_{u,j}\xrightarrow{\mathcal D}\operatorname{Unif}(0,1)\).
Strict increase of the fitted distribution on its positive support
makes the rule \(T_{u,j}>\widehat q_{u,j}(1-\alpha)\) equivalent
to \(p_{u,j}<\alpha\) on \(\mathcal G_n\), proving asymptotic
level \(\alpha\).

The sign conjugation above also proves invariance of the reference
law to all non-target eigenvector signs; the squared inner product
handles the target signs. If \(m_1=m_2=1\), the leading vectors
\(\bl^{(1)},\bl^{(2)}\) are absent. The reduction
through \eqref{supp:eq:two-sample-direction-quadratic} then gives
\(T_{u,1}=o_{\Pp}(1)\), as stated in the main text, but the
nondegenerate calibration does not apply.
\end{proof}

\begin{proof}[Justification of Remark~\ref{rem:estimated-factor-numbers}]
Fix a target \(j\) covered by the corresponding theorem and let
\(\mathcal E_m=\{(\widehat m_1,\widehat m_2)=(m_1,m_2)\}\).
For either test and a fixed \(\alpha\in(0,1)\), let
\(\mathcal R_j^{\mathrm{sel}}(\alpha)\) and \(\mathcal R_j(\alpha)\)
denote the rejection events under the corresponding critical-value rule
with estimated and true factor numbers, respectively.
Using the same folds and other tuning choices, the two procedures
coincide on \(\mathcal E_m\). Hence
\[
 \left|\Pp\{\mathcal R_j^{\mathrm{sel}}(\alpha)\}
       -\Pp\{\mathcal R_j(\alpha)\}\right|
 \le\Pp(\mathcal E_m^c)\longrightarrow0.
\]
Under the respective null hypothesis, the rejection probability with
true factor numbers converges to \(\alpha\), so the same holds with
estimated factor numbers. No independence between rank selection and
testing is required.
\end{proof}

\subsection{Proofs of auxiliary lemmas}
\label{supp:subsec:auxiliary-lemma-proofs}

\begin{proof}[Proof of Lemma~\ref{supp:lem:angular-denominators}]
For a unit vector \(\bv\) and a positive integer \(q\), the Gaussian
representation \(\bs=\bg/\|\bg\|_2\), where
\(\bg\sim N_d(0,\bI_d)\), gives
\begin{equation}\label{supp:eq:spherical-coordinate-moments}
 \E(\bv^\top\bs)^{2q}
 =\frac{(2q-1)!!}{d(d+2)\cdots(d+2q-2)}.
\end{equation}
Indeed, \(\bs\) is independent of \(\|\bg\|_2\), and
\(\E\|\bg\|_2^{2q}=d(d+2)\cdots(d+2q-2)\).
The same argument with the Gaussian fourth moments gives
\begin{equation}\label{supp:eq:spherical-fourth-tensor}
 \E(s_as_bs_cs_e)
 =\frac{\ind(a=b)\ind(c=e)+\ind(a=c)\ind(b=e)
              +\ind(a=e)\ind(b=c)}{d(d+2)}.
\end{equation}
Both identities hold in every orthonormal basis.

Put \(W_a=\bu_a^\top\bs\) and \(p_a=\lambda_a/d\).
Since \(p_a\ge0\) and \(\sum_a p_a=1\), Jensen's inequality gives
\[
 \E\mathcal A^q
 =d^q\E\left(\sum_{a=1}^d p_aW_a^2\right)^q
 \le d^q\sum_{a=1}^d p_a\E W_a^{2q}
 \le(2q-1)!!\le(2q)^q.
\]
Also, \(\mathcal A\ge c_{\varepsilon}\) and \(\E\mathcal A=1\). Independence of
\(\xi\) and \(\bs\) yields
\[
 \zeta_1=\E(\xi^{-1})\E(\mathcal A^{-1/2}),\qquad
 1\le\E(\mathcal A^{-1/2})\le c_{\varepsilon}^{-1/2}.
\]
The inverse-radial condition therefore implies
\[
 \{\E\omega^q\}^{1/q}
 =\frac{\{\E(\xi^{-q})\}^{1/q}}
        {\E(\xi^{-1})\E(\mathcal A^{-1/2})}
 \le K_\xi\sqrt q,\qquad
 1\le q\le\lfloor c_\xi d\rfloor.
\]
In particular, \(\omega\) is independent of \(\mathcal A\), and all
fixed mixed moments needed below are bounded.
The angular moment bound holds for every positive integer \(q\),
without a dimension-dependent upper limit. Take
\(q_n=\lceil c\log n\rceil\) for a fixed \(c>1\). A union bound
and Markov's inequality give
\[
 \Pp\{\max_{i\le n}\mathcal A_i>M\log n\}
 \le n\left(\frac{2q_n}{M\log n}\right)^{q_n}\longrightarrow0
\]
for sufficiently large fixed \(M\). For the joint angular and
inverse-radial maximum, only fixed moments are needed: since \(d\to\infty\),
\(4\le\lfloor c_\xi d\rfloor\) eventually and
\(\E\omega^4\le C\). Independence and the angular moment bound give
\[
 \E\{(1+\mathcal A^{1/2})^4\omega^4\}
 =\E(1+\mathcal A^{1/2})^4\E\omega^4\le C.
\]
Consequently, for every \(M>0\),
\[
 \Pp\{\max_{i\le n}(1+\mathcal A_i^{1/2})\omega_i>M n^{1/4}\}
 \le\sum_{i=1}^n
 \frac{\E\{(1+\mathcal A_i^{1/2})^4\omega_i^4\}}{M^4n}
 \le\frac{C}{M^4},
\]
which proves the joint maximum bound without a
restriction on the relative growth of \(n\) and \(d\).

For the denominator calculation, write
\[
 \bz_i=\bSigma^{1/2}\bs_i,\qquad
 \bt_{\ell i}=\bdelta_{\mu,-\ell}/\xi_i,\qquad
 b_{\ell i}=\|\bt_{\ell i}\|_2
   =\zeta_1\|\bdelta_{\mu,-\ell}\|_2\omega_i.
\]
Then \(\bz_i^\top\bOmega \bz_i=1\), \(\|\bz_i\|_2^2=\mathcal A_i\), and
\[
 e_{\ell i}
 =\bz_i^\top\bDelta_{P,-\ell}\bz_i
  -2\bz_i^\top(\bOmega+\bDelta_{P,-\ell})\bt_{\ell i}
  +\bt_{\ell i}^\top(\bOmega+\bDelta_{P,-\ell})\bt_{\ell i}.
\]
On an event of probability tending to one,
\(\max_\ell\|\bOmega+\bDelta_{P,-\ell}\|_{\op}\le C\). Hence
\[
 |e_{\ell i}|
 \le \|\bDelta_{P,-\ell}\|_{\op}\mathcal A_i
       +C\mathcal A_i^{1/2}b_{\ell i}+Cb_{\ell i}^2.
\]
The nuisance rates and the joint maximum bound show that
\[
 \max_{\ell\le K}\max_{i\in I_\ell}
 (1+\mathcal A_i^{1/2})b_{\ell i}
 =O_{\Pp}(a_n n^{1/4})=o_{\Pp}(1).
\]
Substitution proves \eqref{supp:eq:uniform-denominator-control},
since \(a_n=o(n^{-1/4})\) implies both \(a_n\log n\to0\) and
\(a_n n^{1/4}\to0\).
Finally, conditional on a training sample, a held-out observation
equals its location estimate with probability zero: for each positive
radius, the spherical direction is nonatomic when \(d\ge2\).
Positive definiteness of the pilot thus makes each denominator
strictly positive almost surely, including outside the event used
for the uniform bounds.
\end{proof}

\begin{proof}[Proof of Lemma~\ref{supp:lem:raw-linear}]
We expand location and precision jointly at their population values.
Fix a fold and condition on \(\mathcal F_{-\ell}\). In this proof, write
\(\bDelta=\bDelta_{P,-\ell}\), \(\bh=\bdelta_{\mu,-\ell}\),
\(\delta=\|\bDelta\|_{\op}\), \(\bz=\bSigma^{1/2}\bs\), and
\(\bt=\bh/\xi\), with \(b=\|\bt\|_2\) and \(\mathcal A=\|\bz\|_2^2\).
The rescaled observation-level score is
\[
 S(\bz,\bt,\bDelta)
 =d\,\frac{(\bz-\bt)(\bz-\bt)^\top}{(\bz-\bt)^\top(\bOmega+\bDelta)(\bz-\bt)}.
\]
Write its denominator as \(1+e_1+e_2\), where
\[
 e_1=\bz^\top\bDelta \bz-2\bz^\top\bOmega \bt,\qquad
 e_2=\bt^\top\bOmega \bt-2\bz^\top\bDelta \bt+\bt^\top\bDelta \bt,
 \qquad e=e_1+e_2.
\]
Let \(\bN_0=\bz\bz^\top\), \(\bN_1=-\bz\bt^\top-\bt\bz^\top\), and
\(\bN_2=\bt\bt^\top\).
The identity \((1+e)^{-1}=1-e+e^2/(1+e)\) gives the exact expansion
\begin{equation}\label{supp:eq:joint-score-expansion}
 S(\bz,\bt,\bDelta)=d\bz\bz^\top+L_P(\bz,\bDelta)+L_\mu(\bz,\bt)
                    +\mathcal R(\bz,\bt,\bDelta),
\end{equation}
where
\[
 \begin{aligned}
 L_P(\bz,\bDelta)&=-d\bz\bz^\top(\bz^\top\bDelta \bz),\\
 L_\mu(\bz,\bt)&=d\{-\bz\bt^\top-\bt\bz^\top
                       +2\bz\bz^\top(\bz^\top\bOmega \bt)\},\\
 d^{-1}\mathcal R(\bz,\bt,\bDelta)
 &=\bN_2-\bN_0e_2-(\bN_1+\bN_2)e
                  +(\bN_0+\bN_1+\bN_2)\frac{e^2}{1+e}.
 \end{aligned}
\]
In particular, the expansion uses no inverse involving an estimated
matrix in its linear terms.

On \(\{|e|\le1/2,\ \delta\le1,\ b\le1\}\),
\[
 |e|\le C\{\delta\mathcal A+\mathcal A^{1/2}b+b^2\},
 \qquad |e_2|\le C(b^2+\delta\mathcal A^{1/2}b).
\]
Applying the rank-one bound following
\eqref{supp:eq:joint-matrix-norm} to each term in the exact remainder
therefore gives
\begin{equation}\label{supp:eq:joint-score-remainder-bound}
 \|\mathcal R(\bz,\bt,\bDelta)\|_\diamond
 \le C(1+\mathcal A)^3(\delta+b)^2.
\end{equation}
For example, the last term is bounded by
\[
 C(\mathcal A+\mathcal A^{1/2}b+b^2)
       (\delta\mathcal A+\mathcal A^{1/2}b+b^2)^2
 \le C(1+\mathcal A)^3(\delta+b)^2;
\]
the other terms satisfy the same bound.
Lemma~\ref{supp:lem:angular-denominators} ensures that the event
required for \eqref{supp:eq:joint-score-remainder-bound} holds for
all held-out observations with probability tending to one.
Since \(b=\zeta_1\|\bh\|_2\omega\) and \(\omega\) is independent of
\(\mathcal A\), the same lemma gives, conditionally on
\(\mathcal F_{-\ell}\),
\[
 \E\{(1+\mathcal A)^3(\delta+b)^2\mid\mathcal F_{-\ell}\}
 \le C(\delta+\zeta_1\|\bh\|_2)^2.
\]
Conditional Markov's inequality, localized to
\(\delta\le M a_n\) and \(\zeta_1\|\bh\|_2\le M a_n\),
then proves
\begin{equation}\label{supp:eq:empirical-score-remainder}
 \frac1{N_\ell}\sum_{i\in I_\ell}
       \|\mathcal R(\bz_i,\bt_{\ell i},\bDelta)\|_\diamond
 =O_{\Pp}(a_n^2).
\end{equation}
Taking \(M\) sufficiently large and then \(n,d\) large removes the
localization. The maximum bounds are used only to justify the
expansion; the averaged remainder rate follows from the fixed-moment
bound above. This argument bounds a polynomial majorant on the
high-probability denominator event, so it does not assume moment
bounds for the reciprocal of an estimated denominator on its complement.

Next consider the linear terms. Applying
\eqref{supp:eq:spherical-fourth-tensor} with
\(\bC_{\bDelta}=\bSigma^{1/2}\bDelta\bSigma^{1/2}\) gives
\[
 \E\{\bs\bs^\top(\bs^\top \bC_{\bDelta}\bs)\}
 =\frac{2\bC_{\bDelta}+\tr(\bC_{\bDelta})\bI_d}{d(d+2)}.
\]
Consequently,
\[
 \E(L_P\mid\mathcal F_{-\ell})
 =-\frac{2\bSigma\bDelta\bSigma
                 +\tr(\bSigma\bDelta)\bSigma}{d+2}.
\]
The location term is odd under \(\bs\mapsto-\bs\) for each fixed
\(\xi\); hence \(\E(L_\mu\mid\mathcal F_{-\ell})=0\).
Moreover,
\[
 \|L_P\|_\diamond\le C\delta\mathcal A^2,\qquad
 \|L_\mu\|_\diamond
 \le Cb(\mathcal A^{1/2}+\mathcal A^{3/2}).
\]
The angular and inverse-radial moment bounds imply
\begin{equation}\label{supp:eq:linear-score-second-moments}
 \E(\|L_P\|_\diamond^2\mid\mathcal F_{-\ell})\le C\delta^2,\qquad
 \E(\|L_\mu\|_\diamond^2\mid\mathcal F_{-\ell})
       \le C\zeta_1^2\|\bh\|_2^2.
\end{equation}
Because \(\|\cdot\|_\diamond\) is a Hilbert norm, conditional
independence and centering give, for either linear term \(L\),
\[
 \E\left[
  \left\|\frac1{N_\ell}\sum_{i\in I_\ell}
          \{L_i-\E(L_i\mid\mathcal F_{-\ell})\}\right\|_\diamond^2
       \,\middle|\,\mathcal F_{-\ell}\right]
 \le \frac1{N_\ell}\E(\|L\|_\diamond^2\mid\mathcal F_{-\ell}).
\]
Thus the empirical precision and location linear terms differ from
their conditional means by \(O_{\Pp}(a_n n^{-1/2})\) each.
Combining these bounds with
\eqref{supp:eq:empirical-score-remainder} and summing over the finitely
many folds proves \eqref{supp:eq:raw-matrix-expansion}, since
\(a_n^2+a_n n^{-1/2}\le\rho_n\).

For completeness, the conditional mean identities also give the
population derivatives stated in the lemma. For precision, at each
fixed \(d\), a sufficiently short line segment from \(\bOmega\)
has its relative denominator uniformly bounded away from zero, so
differentiation under the expectation is valid. For location,
truncate to \(\|\by\|_2>2\|\bh\|_2\), where the derivative is bounded
by \(C_d/\|\by\|_2\). On the complementary ball, boundedness of the
Tyler score and \(\E\|\by\|_2^{-2}<\infty\) give an
\(o(\|\bh\|_2)\) contribution to the expectation of the score
difference. The derivative on the truncated region converges by
dominated convergence, giving the zero location derivative.

Finally, \(d\bz\bz^\top\) has mean \(\bSigma\) and
\(\|d\bz\bz^\top\|_\diamond^2=2\mathcal A^2\), so
\[
 \left\|\frac1n\sum_{i=1}^n\bPhi(\bx_i)\right\|_\diamond
 =O_{\Pp}(n^{-1/2}).
\]
The norm of the derivative in \eqref{supp:eq:raw-tyler-derivative} is at most
\(C\|\bH\|_{\op}\), using \(\tr(\bSigma)=d\),
\(\|\bSigma\|_{\mathrm F}\le d\), and
\(\tr(\bSigma^2)\le d^2\).
Adding back the linear precision term proves the final rate in
Lemma~\ref{supp:lem:raw-linear}.
\end{proof}

\begin{proof}[Proof of Lemma~\ref{supp:lem:spectral-perturbation}]
Let \(\bP_j^\perp=\bI_d-\bu_j\bu_j^\top\) and write
\(\widetilde{\bu}_j=\alpha\bu_j+\bh\), with
\(\bh=\bP_j^\perp\widetilde{\bu}_j\) and \(\alpha\ge0\).
Weyl's inequality gives
\(|\widetilde\lambda_j-\lambda_j|\le\|\bE\|_{\op}\).
The projected eigen-equation is
\[
 \bh=\alpha \bG_j\bE\bu_j
       +\bG_j\{\bE-(\widetilde\lambda_j-\lambda_j)\bI_d\}\bh.
\]
Since \(\|\bG_j\|_{\op}=g_j^{-1}\), the event
\(\|\bE\|_{\op}\le g_j/4\) implies
\[
 \|\bh\|_2\le2\|\bG_j\bE\bu_j\|_2
             \le2\|\bE\|_{\op}/g_j\le\tfrac12,\qquad
 \alpha=\sqrt{1-\|\bh\|_2^2}\ge\sqrt3/2.
\]
Projection onto \(\bu_j\) yields
\[
 \widetilde\lambda_j-\lambda_j-\bu_j^\top\bE\bu_j
 =\alpha^{-1}\bu_j^\top\bE\bh,
\]
which proves \eqref{supp:eq:eigenvalue-perturbation}.
Also, \(1-\alpha=\|\bh\|_2^2/(1+\alpha)\), and hence
\[
 \begin{aligned}
 \|\widetilde{\bu}_j-\bu_j-\bG_j\bE\bu_j\|_2
 &\le |1-\alpha|(1+\|\bG_j\bE\bu_j\|_2)
      +\frac{2\|\bE\|_{\op}}{g_j}\|\bh\|_2\\
 &\le C\|\bE\|_{\op}^2/g_j^2.
 \end{aligned}
\]
This proves \eqref{supp:eq:eigenvector-perturbation}, including the
longitudinal normalization error.
\end{proof}

\begin{proof}[Proof of Lemma~\ref{supp:lem:moments}]
Write \(W_a=\bu_a^\top\bs\) and
\(\mathcal A=\sum_a\lambda_aW_a^2\).
The formula for \(\bPsi\) gives directly
\[
 \varphi_{\lambda,j}=dW_j^2-\mathcal A,\qquad
 \bu_k^\top\bq_j
 =\frac{(d+2)\sqrt{\lambda_j\lambda_k}}{\lambda_j-\lambda_k}
        W_jW_k\quad(k\ne j).
\]
Equation~\eqref{supp:eq:spherical-fourth-tensor} implies
\[
 \E W_a^2=\frac1d,\qquad
 \E(W_a^2W_b^2)=\frac{1+2\ind(a=b)}{d(d+2)}.
\]
Together with \(\sum_a\lambda_a=d\), these identities give
\[
 \begin{aligned}
 \Var(\varphi_{\lambda,j})
 &=\frac{3d^2-2d(d+2\lambda_j)
                    +d^2+2\sum_a\lambda_a^2}{d(d+2)}\\
 &=\frac{2}{d(d+2)}
       \left(d^2-2d\lambda_j+\sum_a\lambda_a^2\right).
 \end{aligned}
\]
Multiplying by \(\{(d+2)/d\}^2\) proves the eigenvalue variance
formula. The same fourth moments give the directional variances.
For distinct \(k,l\ne j\), the moment
\(\E(W_j^2W_kW_l)\) is zero by reflection symmetry, proving the
vanishing cross-covariances.

The bounded bulk eigenvalues imply
\(d^{-2}\sum_{a>m}\lambda_a^2=O(d^{-1})\).
It follows that
\[
 \sigma_{\lambda,j}^2\longrightarrow
 2\left(1-2\kappa_j+\sum_{a=1}^m\kappa_a^2\right)
 \ge2(1-\kappa_j)^2>0.
\]
For the last inequality, the bulk lower bound and trace constraint
give \(\sum_{a=1}^m\kappa_a\le1-c_{\varepsilon}\); hence \(\kappa_j<1\).
Equation~\eqref{supp:eq:spherical-coordinate-moments} gives
\(\E(dW_j^2)^4\le105\), and
Lemma~\ref{supp:lem:angular-denominators} gives
\(\E\mathcal A^4\le105\). Thus
\(\E|\varphi_{\lambda,j}|^4\le C\).

For the directional moments, the squared coefficient of \(W_jW_k\)
in \(\bq_j\) is \(O(d^2)\) for \(k\le m,\ k\ne j\),
and \(O(d)\) uniformly for \(k>m\). Consequently,
\[
 \|\bq_j\|_2^2
 \le C\left\{\sum_{\substack{k\le m\\k\ne j}}
                      (dW_jW_k)^2+dW_j^2\right\}.
\]
There are finitely many leading terms, and
\[
 d^4\E(W_j^4W_k^4)
 \le d^4\{\E W_j^8\,\E W_k^8\}^{1/2}\le105.
\]
These bounds prove
\(\E\|\bq_j\|_2^4\le C\).
They also imply
\(\sum_{k\ne j}\sigma_{u,jk}^2=O(1)\).

It remains to prove bulk concentration. Set
\(\bv_i=\bPi_{\mathrm{bulk}}\bq_j(\bx_i)\) and
\(\bC_{\mathrm{bulk}}=\Cov(\bv_i)\). The covariance calculation above gives
\[
 \bC_{\mathrm{bulk}}=\sum_{k>m}\sigma_{u,jk}^2\bu_k\bu_k^\top,\qquad
 \tr \bC_{\mathrm{bulk}}=b_{u,j}=O(1),\qquad
 \tr(\bC_{\mathrm{bulk}}^2)=\sum_{k>m}\sigma_{u,jk}^4=O(d^{-1}),
\]
since \(\max_{k>m}\sigma_{u,jk}^2=O(d^{-1})\).
For independent centered \(\bv_i\),
\[
 \left\|\frac1{\sqrt n}\sum_{i=1}^n\bv_i\right\|_2^2-b_{u,j}
 =\frac1n\sum_{i=1}^n(\|\bv_i\|_2^2-\E\|\bv_i\|_2^2)
   +\frac2n\sum_{i<h}\bv_i^\top\bv_h.
\]
The two terms are uncorrelated, as are the terms corresponding to
distinct unordered pairs in the second sum. Therefore
\[
 \begin{aligned}
 \Var\left\{\left\|\frac1{\sqrt n}\sum_{i=1}^n\bv_i\right\|_2^2\right\}
 &=\frac1n\Var(\|\bv_1\|_2^2)
       +\frac{2(n-1)}n\tr(\bC_{\mathrm{bulk}}^2)\\
 &\le\frac1n\E\|\bv_1\|_2^4+2\tr(\bC_{\mathrm{bulk}}^2)
 =O(n^{-1}+d^{-1}).
 \end{aligned}
\]
Chebyshev's inequality proves \eqref{supp:eq:bulk-concentration}.
\end{proof}

\begin{proof}[Proof of Lemma~\ref{supp:lem:quadratic-cdf}]
Diagonalizing \(\bK\), write
\(\bg^\top \bK\bg=\lambda g_1^2+R\), where
\(\lambda=\lambda_{\max}(\bK)\ge c\), \(R\ge0\), and \(R\) is
independent of \(g_1\). The density of \(g_1^2\) on \((0,\infty)\)
is \((2\pi x)^{-1/2}e^{-x/2}\), which is decreasing. Thus any
interval of length \(t\) has probability at most
\[
 \Pp(\lambda g_1^2\le t)
 \le\sqrt{\frac{2t}{\pi\lambda}}
 \le C\sqrt t.
\]
Conditioning on \(R\) proves the first bound in
\eqref{supp:eq:quadratic-cdf-bounds}, and hence continuity.
At least one weight is positive, so there is no mass at zero.
Independence of the entries of the standard Gaussian vector \(\bg\)
also gives positive
probability to every interval in \((0,\infty)\); this proves
strict increase on \([0,\infty)\).

For the second bound, couple both quadratic forms using the same
standard Gaussian vector. If \(\delta=\|\bK-\bL\|_{\op}\), then
\[
 |\bg^\top(\bK-\bL)\bg|\le\delta\|\bg\|_2^2,\qquad
 \Pp\{|\bg^\top(\bK-\bL)\bg|>t\}\le q\delta/t.
\]
The event inclusions for two random variables whose difference is
at most \(t\), together with the first bound, imply
\[
 \sup_x|F_{\bK}(x)-F_{\bL}(x)|\le C\sqrt t+q\delta/t.
\]
For \(\delta>0\), choose \(t=\delta^{2/3}\); for \(\delta=0\),
the laws are identical. This proves the claimed bound, without a
positive lower bound on the other eigenvalues of either matrix.
\end{proof}

\section{Additional Simulation Studies}
\label{supp:sec:simulations}

\subsection{Implementation details}
\label{supp:subsec:simulation-implementation}

We implement the elliptical model in Section~\ref{subsec:simulation-design}
using an equivalent Gaussian scale-mixture representation. For each
group \(r=1,2\), we generate
\(\bg_i^{(r)}\sim N_d(\boldsymbol0,\bSigma^{(r)})\) independently
across observations and groups. For the Gaussian distribution, we set
\(\bx_i^{(r)}=\bg_i^{(r)}\); for the elliptical \(t_\nu\) distribution,
we set
\[
 \bx_i^{(r)}=\sqrt{\frac{\nu-2}{V_i^{(r)}}}\,\bg_i^{(r)},
 \qquad V_i^{(r)}\sim\chi_\nu^2,\qquad
 V_i^{(r)}\perp\bg_i^{(r)},\qquad\nu\in\{8,3\}.
\]
The mixing variables are independent across observations and groups
and independent of all Gaussian vectors.
This gives the radial distributions in Section~\ref{subsec:simulation-design}
because a standard Gaussian vector has independent length and direction,
with squared length distributed as \(\chi_d^2\) and direction uniform
on \(\mathbb S^{d-1}\).

For the proposed method, we use \(K=3\) balanced folds in each group.
Let \(N_\ell^{(r)}\) be the size of evaluation fold \(\ell\), so that
the corresponding training sample has size
\(N_{-\ell}^{(r)}=n_r-N_\ell^{(r)}\), for \(\ell=1,2,3\).
On each training sample, we compute the spatial median by Weiszfeld
iteration, initialized at the coordinatewise median, with relative
tolerance \(10^{-6}\) and at most 100 iterations. We then construct the
POET--SS precision pilot as in Section~\ref{subsec:pilot-construction},
using hard thresholding of the off-diagonal residual entries at level
\[
 0.60\left\{
 \sqrt{\frac{\log d}{N_{-\ell}^{(r)}}}
 +\sqrt{\frac{\log N_{-\ell}^{(r)}}{N_{-\ell}^{(r)}}}
 +d^{-1/2}\right\}.
\]
The residual diagonal is unchanged, and we apply the positive-definite
adjustment in \eqref{eq:foldwise-poet-ss} before inversion. The threshold
constant \(C=0.60\) is fixed across all designs. The shape estimator and
tests are then computed as in Sections~\ref{subsec:orth-estimator}
and~\ref{sec:two-sample-tests}.

We evaluate Gaussian quadratic-form probabilities and quantiles using
Imhof's method \citep{supp:Imhof1961}, with absolute and relative tolerances
\(10^{-8}\). For the sample-covariance benchmark, we use the uncentered
sample second-moment matrix because the simulated populations have zero
location. The master random-number seed is 20260801.

For the power experiments, we use the alternatives defined in
Section~\ref{subsec:simulation-design}. The eigenvalue alternative varies
the multiplier over
\[
 h\in\{1.05,1.10,1.20,1.35,1.50,1.75,2.00,2.50\}.
\]
The eigendirection alternative varies the rotation angle over
\[
 \theta\in\{3^\circ,6^\circ,9^\circ,12^\circ,15^\circ,
            25^\circ,35^\circ,45^\circ,60^\circ\}.
\]
These changes are applied to the factor component. After adding
\(\bR\) and normalizing the trace, we compute \(\Delta_\lambda\) and
\(\Delta_u\) from the full shape matrices using \eqref{eq:simulation-effects}.
For the full shape matrices, the addition of \(\bR\) means that the
eigenvalue alternative can also change the eigendirections, and the
eigendirection alternative can also change the eigenvalues.

To obtain the size-adjusted power in Figure~\ref{fig:power-adjusted},
we calibrate critical values at level \(0.05\) separately for each
method, test, and distribution. We use 2,000 independent null
replications with \(\bL^{(1)}=\bL^{(2)}=\bL\),
\(d=n_1=n_2=250\), and \(m_1=m_2=3\), matching the settings for
the power experiments. Both methods use three factors in each group,
and the tests concern the first principal component. The resulting
critical values are held fixed across effect sizes and applied to
2,000 alternative replications at each point; the rejection proportion
gives the size-adjusted power. The null and alternative simulations
are independent. Within the alternative simulations, we reuse random
draws across effect sizes for each distribution.

\subsection{Null diagnostics under the Gaussian distribution}
\label{supp:subsec:gaussian-diagnostics}

We examine the null distributions of the statistics under the Gaussian
distribution, using \(d=n_1=n_2=500\), \(m_1=m_2=3\), and 5,000
replications. Both methods use three factors in each group, and we
consider the first principal component. Figure~\ref{supp:fig:qq-normal}
shows the Q--Q plots and histograms, complementing the elliptical
\(t_3\) results in Figure~\ref{fig:qq-t3}.

The statistics of both methods broadly follow their reference
distributions. For the proposed eigenvalue and eigendirection
statistics, the Q--Q slopes are \(1.010\) and \(1.067\), respectively.
The corresponding rejection probabilities are \(0.056\) and \(0.060\),
compared with \(0.051\) and \(0.053\) for the benchmark.
The eigendirection plots use oracle reference distributions evaluated
at the population parameters, with 50,000 quantiles computed by
Imhof's method. The pointwise Q--Q reference envelopes are obtained
from the beta order-statistic law. These plots assess the distributional
approximations; Table~\ref{tab:size-balanced} reports rejection
probabilities when the reference distributions are estimated from the data.

\begin{figure}[tb]
\centering
\includegraphics[width=0.85\textwidth]{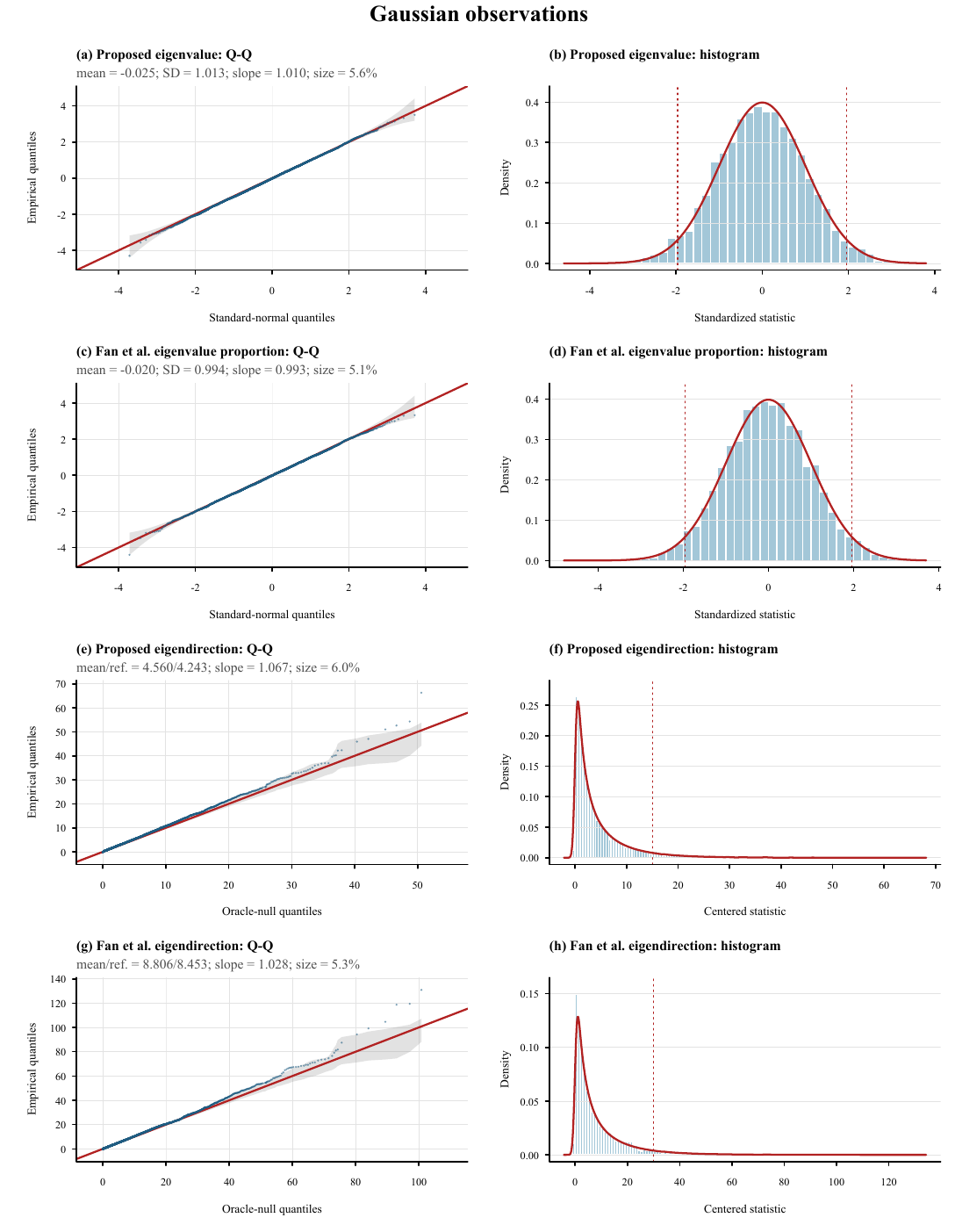}
\caption{Null diagnostics under the Gaussian distribution for the first
principal component, with \(d=n_1=n_2=500\), \(m_1=m_2=3\), and 5,000
replications. In the Q--Q plots, red lines show \(y=x\), and gray bands
are pointwise \(95\%\) reference envelopes for the empirical quantiles.
Each reported Q--Q slope is computed from quantile pairs with probability
levels between \(0.05\) and \(0.95\). In the histograms, red curves show
the reference densities, and dashed vertical lines indicate the critical
values at significance level \(0.05\).}
\label{supp:fig:qq-normal}
\end{figure}

\subsection{Rejection probabilities without size adjustment}
\label{supp:subsec:nominal-power}

We also report rejection probabilities using the critical values from
each method's reference distribution at nominal level \(0.05\).
Figure~\ref{supp:fig:power-nominal} uses the same settings and alternatives
as Figure~\ref{fig:power-adjusted}, with \(d=n_1=n_2=250\) and tests
for the first principal component.

Under the Gaussian distribution, the nominal and size-adjusted curves
are similar. Under the elliptical \(t_3\) distribution, however, the
null rejection probabilities are \(0.058\) and \(0.048\) for the
proposed eigenvalue and eigendirection tests, compared with \(0.008\)
and \(0.650\) for the benchmark. The large rejection probabilities
for the benchmark eigendirection test therefore partly reflect its
inflated null rejection rate. This explains the use of size-adjusted
power in the main text to compare the methods at the same empirical size.

\begin{figure}[tb]
\centering
\includegraphics[width=\textwidth]{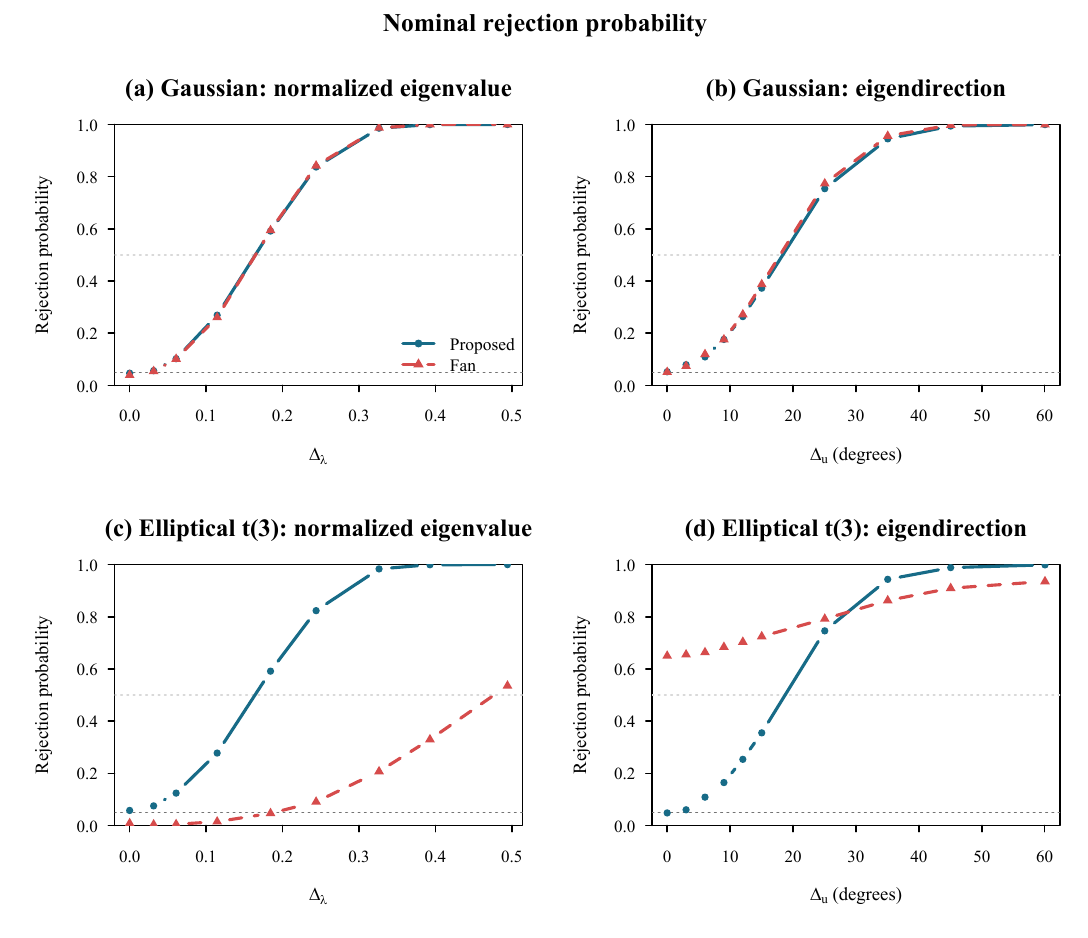}
\caption{Rejection probabilities at nominal level \(0.05\) for the first
principal component under the Gaussian and elliptical \(t_3\) distributions,
with \(d=n_1=n_2=250\), \(m_1=m_2=3\), and 2,000 replications per point.
Blue circles and red triangles denote the proposed method and the
benchmark, respectively. The dark dotted line marks the level \(0.05\).}
\label{supp:fig:power-nominal}
\end{figure}

\subsection{Sensitivity to the fitted factor number}
\label{supp:subsec:factor-sensitivity}

We assess sensitivity to the number of factors used in the tests,
denoted by \(\widetilde m\) for both methods in each group. We keep
the true factor numbers at \(m_1=m_2=3\), set \(d=n_1=n_2=500\), and
consider \(\widetilde m\in\{2,3,4\}\) under the Gaussian, elliptical
\(t_8\), and elliptical \(t_3\) distributions. Each setting uses
5,000 null replications.

Table~\ref{supp:tab:size-sensitivity} shows that retaining one additional
factor has little effect on the proposed rejection probabilities.
In contrast, using only two factors increases them to
\(0.247\)--\(0.406\). Omitting a pervasive factor leaves a strong
component in the POET--SS residual, which is inconsistent with the
sparse residual structure required by the precision pilot.
For the benchmark, changing \(\widetilde m\) affects the reference
distribution but leaves the first two sample-covariance eigenpairs
unchanged. Its rejection probabilities vary little under the Gaussian
distribution, while the distortions under the elliptical \(t_8\) and
\(t_3\) distributions persist.

\begin{table}[tb]
\centering
\footnotesize
\setlength{\tabcolsep}{6.5pt}
\renewcommand{\arraystretch}{1.00}
\begin{threeparttable}
\caption{Empirical rejection probabilities under the null at nominal
level \(0.05\) for different fitted factor numbers, with
\(d=n_1=n_2=500\) and true factor numbers \(m_1=m_2=3\).}
\label{supp:tab:size-sensitivity}
\begin{tabular}{lccrrrr}
\toprule
& & & \multicolumn{2}{c}{Eigenvalue}
& \multicolumn{2}{c}{Eigendirection} \\
\cmidrule(lr){4-5}\cmidrule(lr){6-7}
Distribution & \(\widetilde m\) & \(j\)
& \multicolumn{1}{c}{Proposed} & \multicolumn{1}{c}{Fan et al.}
& \multicolumn{1}{c}{Proposed} & \multicolumn{1}{c}{Fan et al.} \\
\midrule
\multirow{8}{*}{Gaussian} & \multirow{2}{*}{2} & 1 & 0.262 & 0.052 & 0.325 & 0.051 \\
 &  & 2 & 0.261 & 0.052 & 0.394 & 0.063 \\
\addlinespace[1pt]
 & \multirow{3}{*}{3} & 1 & 0.051 & 0.051 & 0.052 & 0.050 \\
 &  & 2 & 0.051 & 0.052 & 0.052 & 0.047 \\
 &  & 3 & 0.047 & 0.044 & 0.058 & 0.051 \\
\addlinespace[1pt]
 & \multirow{3}{*}{4} & 1 & 0.053 & 0.051 & 0.052 & 0.050 \\
 &  & 2 & 0.053 & 0.051 & 0.051 & 0.047 \\
 &  & 3 & 0.050 & 0.044 & 0.058 & 0.050 \\
\midrule
\multirow{8}{*}{\(t_8\)} & \multirow{2}{*}{2} & 1 & 0.255 & 0.033 & 0.327 & 0.121 \\
 &  & 2 & 0.260 & 0.036 & 0.398 & 0.152 \\
\addlinespace[1pt]
 & \multirow{3}{*}{3} & 1 & 0.049 & 0.032 & 0.051 & 0.117 \\
 &  & 2 & 0.051 & 0.034 & 0.049 & 0.126 \\
 &  & 3 & 0.051 & 0.030 & 0.058 & 0.264 \\
\addlinespace[1pt]
 & \multirow{3}{*}{4} & 1 & 0.051 & 0.031 & 0.049 & 0.117 \\
 &  & 2 & 0.052 & 0.034 & 0.049 & 0.126 \\
 &  & 3 & 0.052 & 0.030 & 0.059 & 0.263 \\
\midrule
\multirow{8}{*}{\(t_3\)} & \multirow{2}{*}{2} & 1 & 0.247 & 0.005 & 0.329 & 0.704 \\
 &  & 2 & 0.260 & 0.024 & 0.406 & 0.783 \\
\addlinespace[1pt]
 & \multirow{3}{*}{3} & 1 & 0.050 & 0.004 & 0.054 & 0.698 \\
 &  & 2 & 0.051 & 0.024 & 0.050 & 0.734 \\
 &  & 3 & 0.049 & 0.026 & 0.061 & 0.995 \\
\addlinespace[1pt]
 & \multirow{3}{*}{4} & 1 & 0.051 & 0.004 & 0.054 & 0.698 \\
 &  & 2 & 0.052 & 0.024 & 0.050 & 0.733 \\
 &  & 3 & 0.048 & 0.026 & 0.062 & 0.984 \\
\bottomrule
\end{tabular}
\begin{tablenotes}[flushleft]
\footnotesize
\item \(\widetilde m\) denotes the number of factors used by both methods
in each group. When \(\widetilde m=2\), only the first two principal
components are tested.
\end{tablenotes}
\end{threeparttable}
\end{table}

\subsection{Serial dependence and model misspecification}
\label{supp:subsec:serial-stress}

We further examine departures from temporal independence and the
elliptical distributional assumption, following the simulation design
of \citet{supp:FanLiXiaZheng2026}. In each group \(r=1,2\), we use the
standardized first-order autoregressive process
\[
 \bz_t^{(r)}=0.1\bz_{t-1}^{(r)}
 +\sqrt{1-0.1^2}\,\boldsymbol\varepsilon_t^{(r)},
\]
where the entries of each innovation vector
\(\boldsymbol\varepsilon_t^{(r)}\) are independent standardized \(t_5\)
variables, and the innovation vectors are independent across time.
We set \(d\in\{100,300,500\}\), \(n_1=d\), and \(n_2=1.5d\), and
consider the tests for the first principal component. Both methods use
three factors in each group, with 1,000 replications per setting.
The eigenvalue and eigendirection alternatives use \(h=1.20\) and
\(\theta=9^\circ\), respectively.

\begin{table}[tb]
\centering
\small
\begin{threeparttable}
\caption{Empirical rejection probabilities at nominal level \(0.05\)
for the first principal component in the AR(1)--\(t_5\) experiment,
with \(n_1=d\), \(n_2=1.5d\), and 1,000 replications.}
\label{supp:tab:serial-stress}
\setlength{\tabcolsep}{4.5pt}
\begin{tabular}{r*{8}{c}}
\toprule
& \multicolumn{4}{c}{Null size} & \multicolumn{2}{c}{Eigenvalue alternative} & \multicolumn{2}{c}{Direction alternative} \\ 
\cmidrule(lr){2-5}\cmidrule(lr){6-7}\cmidrule(lr){8-9}
$d$ & Prop. & Fan & Prop. & Fan & Prop. & Fan & Prop. & Fan \\ 
& \multicolumn{2}{c}{Eigenvalue} & \multicolumn{2}{c}{Direction} & & & & \\ 
\midrule
100 & 0.071 & 0.038 & 0.071 & 0.084 & 0.149 & 0.099 & 0.127 & 0.148 \\
300 & 0.074 & 0.036 & 0.126 & 0.145 & 0.373 & 0.262 & 0.273 & 0.298 \\
500 & 0.068 & 0.027 & 0.094 & 0.100 & 0.553 & 0.437 & 0.359 & 0.355 \\
\bottomrule
\end{tabular}
\begin{tablenotes}[flushleft]
\footnotesize
\item Prop. and Fan denote the proposed method and the benchmark,
respectively.
\end{tablenotes}
\end{threeparttable}
\end{table}

Table~\ref{supp:tab:serial-stress} shows size distortions for both
methods. For example, when \(d=300\), the null rejection probabilities
of the proposed and benchmark eigendirection tests are \(0.126\) and
\(0.145\), respectively. Both methods use reference distributions derived
under temporal independence, without a long-run variance adjustment.
Moreover, the independent univariate \(t_5\) innovations differ from
the elliptical construction in Section~\ref{subsec:simulation-design},
where a single random scale multiplies the whole Gaussian vector.
The results therefore reflect departures from both temporal independence
and the elliptical model. The rejection probabilities under the
alternatives should be interpreted alongside these size distortions.

\clearpage
\section{Additional Analysis of S\&P 500 Returns}
\label{supp:sec:empirical}

\subsection{Tail behavior of stock returns}
\label{supp:subsec:empirical-tails}

We further examine the tail behavior of the daily log returns in the
200-stock panel from Section~\ref{sec:application}. Each stock--year
series is centered and standardized separately. For legibility, the
Q--Q plot in Figure~\ref{supp:fig:sp500-tail-qq} uses a fixed random
subsample of at most 200,000 standardized observations; all numerical
summaries use the full sample.

Across the 3,600 stock--year series, the median sample excess kurtosis is
\(2.5\), and the 90th and 95th percentiles are \(10.355\) and \(15.795\),
respectively. In addition, \(43.2\%\) of the series have excess kurtosis
above \(3\). Among the pooled standardized returns, the proportion
with an absolute value exceeding \(4\) is \(0.00424\), compared with
\(2\Phi(-4)\approx0.00006\) under the standard normal distribution.
These summaries and the departures in both tails of the Q--Q plot
indicate heavier empirical tails than under a Gaussian model,
supporting the use of methods that accommodate heavy-tailed returns.

\begin{figure}[tb]
\centering
\includegraphics[width=0.68\textwidth]{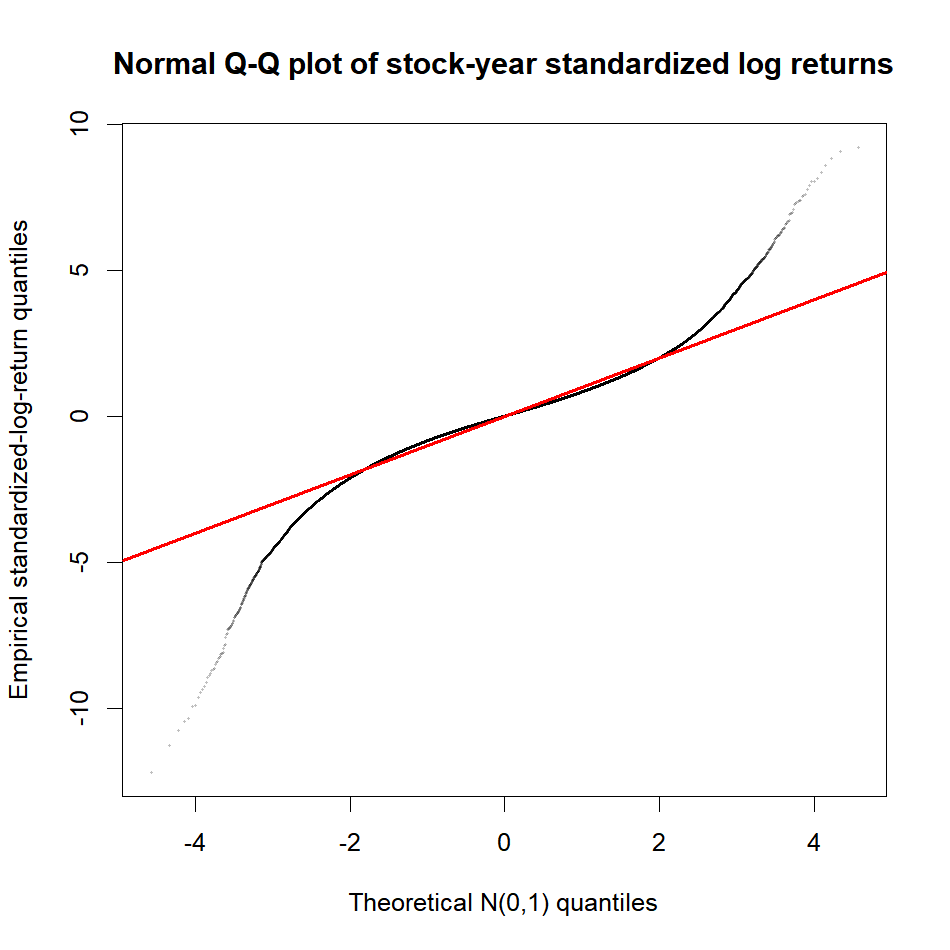}
\caption{Normal Q--Q plot of daily log returns for the fixed 200-stock
S\&P 500 panel, standardized within each stock--year series.
The red line shows \(y=x\).}
\label{supp:fig:sp500-tail-qq}
\end{figure}

\clearpage
\subsection{Comparison with the sample-covariance benchmark}
\label{supp:subsec:empirical-comparison}

Table~\ref{supp:tab:sp500-method-comparison} compares the proposed tests
with those of \citet{supp:FanLiXiaZheng2026}, using the same 200-stock panel,
sample windows, and selected factor numbers as in
Section~\ref{sec:application}. The benchmark results are computed for
our fixed panel; the stock identifiers and sampling seed used in their
published application are unavailable.

At level \(0.05\), the methods give the same decision for 28 of the
36 testable hypotheses concerning eigenvalue or eigendirection equality.
Both reject PC1 eigendirection equality at all seven dates. For PC2,
the eigenvalue decisions differ in June 2009, November 2016, and
October 2022. For example, in November 2016, the proposed eigenvalue
test rejects with a p-value of \(0.017\), while the benchmark does not
reject, with a p-value of \(0.801\). The eigendirection tests give the
opposite decisions, with p-values of \(0.292\) and \(0.001\) for the
proposed method and the benchmark, respectively. This comparison shows
that the choice of method can affect the conclusions for individual
components and dates.

\begin{table}[tb]
\centering
\small
\setlength{\tabcolsep}{4.2pt}
\begin{threeparttable}
\caption{The p-values from the proposed method and the benchmark for
testing the eigenvalues and eigendirections of the fixed 200-stock
S\&P 500 panel at the seven candidate dates.}
\label{supp:tab:sp500-method-comparison}
\begin{tabular}{lcc*{4}{c}}
\toprule
& & & \multicolumn{2}{c}{Eigenvalue} & \multicolumn{2}{c}{Eigendirection} \\
\cmidrule(lr){4-5}\cmidrule(lr){6-7}
Date & $(\widehat m_1,\widehat m_2)$ & \(j\) & Proposed & Fan & Proposed & Fan \\
\midrule
\multirow{3}{*}{2008-09-12} & \multirow{3}{*}{$(4,4)$} & 1 & \(\boldsymbol{<0.001}\) & \(\boldsymbol{<0.001}\) & \(\boldsymbol{<0.001}\) & \(\boldsymbol{<0.001}\) \\
 &  & 2 & \(0.185\) & \(0.352\) & \(\boldsymbol{<0.001}\) & \(\boldsymbol{<0.001}\) \\
 &  & 3 & \(0.639\) & \(0.684\) & \(0.430\) & \(0.074\) \\
\addlinespace[1.5pt]
\multirow{3}{*}{2009-06-22} & \multirow{3}{*}{$(4,2)$} & 1 & \(\mathbf{0.026}\) & \(\mathbf{0.034}\) & \(\boldsymbol{<0.001}\) & \(\boldsymbol{<0.001}\) \\
 &  & 2 & \(\mathbf{0.012}\) & \(0.696\) & \(\boldsymbol{<0.001}\) & \(\boldsymbol{<0.001}\) \\
 &  & 3 & NA & NA & NA & NA \\
\addlinespace[1.5pt]
\multirow{3}{*}{2011-12-21} & \multirow{3}{*}{$(3,2)$} & 1 & \(\boldsymbol{<0.001}\) & \(\boldsymbol{<0.001}\) & \(\boldsymbol{<0.001}\) & \(\boldsymbol{<0.001}\) \\
 &  & 2 & \(\mathbf{0.010}\) & \(\mathbf{0.002}\) & \(\boldsymbol{<0.001}\) & \(\boldsymbol{<0.001}\) \\
 &  & 3 & NA & NA & NA & NA \\
\addlinespace[1.5pt]
\multirow{3}{*}{2014-10-07} & \multirow{3}{*}{$(2,4)$} & 1 & \(\mathbf{0.001}\) & \(\mathbf{0.008}\) & \(\boldsymbol{<0.001}\) & \(\boldsymbol{<0.001}\) \\
 &  & 2 & \(\boldsymbol{<0.001}\) & \(\mathbf{0.004}\) & \(\boldsymbol{<0.001}\) & \(\boldsymbol{<0.001}\) \\
 &  & 3 & NA & NA & NA & NA \\
\addlinespace[1.5pt]
\multirow{3}{*}{2016-11-07} & \multirow{3}{*}{$(4,4)$} & 1 & \(\boldsymbol{<0.001}\) & \(\boldsymbol{<0.001}\) & \(\boldsymbol{<0.001}\) & \(\boldsymbol{<0.001}\) \\
 &  & 2 & \(\mathbf{0.017}\) & \(0.801\) & \(0.292\) & \(\mathbf{0.001}\) \\
 &  & 3 & \(0.364\) & \(0.433\) & \(0.350\) & \(\mathbf{0.001}\) \\
\addlinespace[1.5pt]
\multirow{3}{*}{2020-02-21} & \multirow{3}{*}{$(4,7)$} & 1 & \(\boldsymbol{<0.001}\) & \(\boldsymbol{<0.001}\) & \(\boldsymbol{<0.001}\) & \(\boldsymbol{<0.001}\) \\
 &  & 2 & \(0.228\) & \(0.826\) & \(\boldsymbol{<0.001}\) & \(\boldsymbol{<0.001}\) \\
 &  & 3 & \(0.583\) & \(0.944\) & \(\mathbf{0.012}\) & \(0.164\) \\
\addlinespace[1.5pt]
\multirow{3}{*}{2022-10-17} & \multirow{3}{*}{$(6,5)$} & 1 & \(\mathbf{0.035}\) & \(0.065\) & \(\boldsymbol{<0.001}\) & \(\boldsymbol{<0.001}\) \\
 &  & 2 & \(\mathbf{0.005}\) & \(0.103\) & \(0.152\) & \(\mathbf{0.027}\) \\
 &  & 3 & \(0.961\) & \(0.989\) & \(0.482\) & \(0.366\) \\
\bottomrule
\end{tabular}
\begin{tablenotes}[flushleft]\footnotesize
\item Fan denotes the sample-covariance benchmark.
Bold entries indicate rejection at \(0.05\);
NA denotes \(j>\min(\widehat m_1,\widehat m_2)\).
\end{tablenotes}
\end{threeparttable}
\end{table}


\end{document}